\documentclass[11pt]{article}

\usepackage[colorlinks=true,citecolor=blue]{hyperref}
\usepackage{natbib}
\usepackage{natbib,extarrows}
\usepackage{float}
\usepackage{dsfont}
\usepackage{fancyhdr}
\usepackage{indentfirst}
\usepackage{subcaption}
\usepackage{titlesec}
\usepackage{url}
\usepackage[misc]{ifsym}
\usepackage[T1]{fontenc}
\usepackage{soul}            % For \st{} (i.e. strikethrough) only

\usepackage{amsfonts}
\usepackage{amsmath}
\usepackage{amsthm}
\usepackage{amssymb}
\usepackage{cases}
\usepackage{enumerate}
\usepackage{graphicx}
\theoremstyle{definition}

\newtheorem{definition}{Definition}

\theoremstyle{plain}
\newtheorem{theorem}{Theorem}
\newtheorem{lemma}{Lemma}
\newtheorem{proposition}{Proposition}
\newtheorem{corollary}{Corollary}
\theoremstyle{remark}
\newtheorem{remark}{Remark}
    \allowdisplaybreaks

\theoremstyle{definition}
\def\Fbar{ {\overline F}}         
\def\oo{\infty}                   
\def\lm{\lambda}                  \def\th{\theta}
\def\sg{{\sigma}}

\def\E{\mathbb{E}}
\def\N{\mathbb{N}}
\def\P{\mathbb{P}}

\def\R{\mathbb{R}}

\def\Cov{\mathrm{Cov}}

\def\Var{\mathrm{Var}}

\newcommand{\sB}[1]{\left[ #1 \right]}
\newcommand{\rB}[1]{\left( #1 \right)}
\newcommand{\cB}[1]{\left\{ #1 \right\}}

\newcommand{\iSum}{\sum^n_{i=1}}

\newcommand{\INTx}[1]{\int^{\oo}_{0} #1 \,dx}

\newcommand{\INTth}[1]{\int^{\oo}_{0} #1 \,d\th}

\newcommand{\al}{\alpha}
\newcommand{\Th}{\Theta}
\newcommand{\gm}{\gamma}

\newcommand{\aCTEa}[1]{\mathrm{CTE}_\al(#1)}
\newcommand{\bTMa}[2]{\mathrm{TM}_{\al, #1}(#2)}
\newcommand{\bTCMa}[2]{\mathrm{TCM}_{\al, #1}(#2)}

\newcommand{\aINTSa}[1]{\int^{\oo}_{s_\al} #1 \,ds}

\newcommand{\SUMj}{\sum^n_{j=1}}
\newcommand{\SUMk}{\sum^n_{k=1}}

\newcommand{\SL}{S^{*(l)}}
\newcommand{\Cl}{c^{*(l)}}
\newcommand{\ALl}{\al^{*(l)}}
\newcommand{\SLl}{S^{*(l+1)}}
\newcommand{\Cll}{c^{*(l+1)}}
\newcommand{\ALll}{\al^{*(l+1)}}
\newcommand{\FRl}{\frac{(1-\ALll)\Cll}{(1-\ALl)\Cl}}

\newcommand{\aFRACa}[1]{\frac{#1}{1-\al}}
\newcommand{\FRa}{\frac{1}{1-\al}}
\newcommand{\FRb}{\aFRACa{1-\ALa}}
\newcommand{\FRc}{\aFRACa{1-\ALa} \Ca}
\newcommand{\FRd}{\aFRACa{1-\ALb} \Cb}
\newcommand{\ALa}{\al^*}
\newcommand{\ALb}{\al^{**}}

\newcommand{\SGi}{\sg^2_i}

\newcommand{\SGS}{\sg^2_S}
\newcommand{\SGij}{\sg_{ij}}

\newcommand{\SGiS}{\sg_{iS}}
\newcommand{\SGjS}{\sg_{jS}}

\newcommand{\MUi}{\mu_i}

\newcommand{\MUk}{\mu_k}
\newcommand{\GMi}{\gm_i}
\newcommand{\GMj}{\gm_j}
\newcommand{\GMk}{\gm_k}

\newcommand{\AIa}{a_{0,i}}
\newcommand{\AIb}{a_{1,i}}
\newcommand{\AIc}{a_{2,i}}
\newcommand{\AJa}{a_{0,j}}
\newcommand{\AJb}{a_{1,j}}
\newcommand{\AJc}{a_{2,j}}

\newcommand{\Ca}{c^*}
\newcommand{\Cb}{c^{**}}

\newcommand{\ST}{S \mid \Th}
\newcommand{\St}{S \mid \th}

\newcommand{\FBARSt}{\Fbar_{S \mid \th}(s_\al)}

\newcommand{\FS}{f_S(s)}
\newcommand{\FSt}{f_{\St}(s)}
\newcommand{\aHS}[1]{h_{S^{#1}}(s_\al)}

\newcommand{\PIt}{\pi(\th)}
\newcommand{\PIta}{\pi^*(\th)}
\newcommand{\PItb}{\pi^{**}(\th)}

\newcommand{\cTCMb}[3]{\bCES{#1 \rB{#2 - \aCTEa{#2}}^{#3}}{#2}}

\newcommand{\aCEEs}[1]{\E \sB{#1 \mid S=s}}
\newcommand{\aCEEst}[1]{\E \sB{#1 \mid S=s, \Th = \th}}
\newcommand{\aCOVst}[1]{\Cov \sB{#1 \mid S=s, \Th = \th}}

\newcommand{\bCES}[2]{\E \sB{#1 \mid #2 > s_\al}}
\newcommand{\bCESt}[2]{\E \sB{#1 \mid #2 > s_\al, \Th = \th}}
\newcommand{\aTCOVS}[1]{\Cov \sB{#1 \mid S > s_\al}}
\newcommand{\aTCMS}[1]{\E \sB{\rB{S - \aCTEa{S}}^{#1} \mid S>s_\al}}

\usepackage[onehalfspacing]{setspace}

\usepackage{bm}
\usepackage{tikz-qtree}
\usepackage{tikz}

 \pgfdeclarepatternformonly{south east lines}{\pgfqpoint{-0pt}{-0pt}}{\pgfqpoint{3pt}{3pt}}{\pgfqpoint{3pt}{3pt}}{
    \pgfsetlinewidth{0.4pt}
    \pgfpathmoveto{\pgfqpoint{0pt}{3pt}}
    \pgfpathlineto{\pgfqpoint{3pt}{0pt}}
    \pgfpathmoveto{\pgfqpoint{.2pt}{-.2pt}}
    \pgfpathlineto{\pgfqpoint{-.2pt}{.2pt}}
    \pgfpathmoveto{\pgfqpoint{3.2pt}{2.8pt}}
    \pgfpathlineto{\pgfqpoint{2.8pt}{3.2pt}}
    \pgfusepath{stroke}}
\begin{document}

\title{Capital allocation and tail central moments for the multivariate normal mean-variance mixture distribution}

\author{
Enrique Calder\'{i}n-Ojeda
\thanks{Department of Economics, University of Melbourne, Australia. 
\Letter~{\scriptsize\url{enrique.calderin@unimelb.edu.au}}}
\and
Yuyu Chen
\thanks{Department of Economics, University of Melbourne, Australia. 
\Letter~{\scriptsize\url{yuyu.chen@unimelb.edu.au}}}
\and 
Soon Wei Tan % \textsuperscript{\S,}
\thanks{Department of Economics, University of Melbourne, Australia. \Letter~{\scriptsize\url{soonweit@student.unimelb.edu.au}}}
\thanks{Corresponding author}
}

% {\renewcommand\thefootnote{}
% \footnotetext{\textsuperscript{\S} Corresponding author.}
% }
\date{}
\maketitle

\begin{abstract} 
Capital allocation is a procedure used to assess the risk contributions of individual risk components to the total risk of a portfolio. While the conditional tail expectation (CTE)-based capital allocation is arguably the most popular capital allocation method, its inability to reflect important tail behaviour of losses necessitates a more accurate approach. In this paper, we introduce a new capital allocation method based on the tail central moments (TCM), generalising the tail covariance allocation informed by the tail variance. We develop analytical expressions of the TCM as well as the TCM-based capital allocation for the class of normal mean-variance mixture distributions, which is widely used to model asymmetric and heavy-tailed data in finance and insurance. As demonstrated by a numerical analysis, the TCM-based capital allocation captures several significant patterns in the tail region of equity losses that remain undetected by the CTE, enhancing the understanding of the tail risk contributions of risk components. 
 
\end{abstract}

\section{Introduction}
Risk assessment is a core task in finance and insurance. For an agent who manages a portfolio consisting of multiple assets, a common procedure is capital allocation. This is usually achieved through two main steps. Firstly, the agent decides on a total capital reserve based on their risk preferences. Secondly, the capital reserve is distributed across all individual assets in a way that reflects their risk contributions. Capital allocation has broader purposes than its literal meaning of physically allocating capital to each asset, such as deciding portfolio weights, comparing asset profitability, and so on. For discussions on various capital allocation principles, properties, and applications, see, e.g., \cite{D01}, \cite{K05}, \cite{D12}, \cite{G21}, and references therein.
% (From: Covariance Principle for Capital Allocation: A Time-Varying Approach) 
% 1. (Applications of Capital Allocation) Examples of capital allocation problems can be found in many areas of finance, where agents attempt to define investment risks [5–11], and in insurance, where they involve the allocation of the total solvency capital requirement across the business lines or the total costs or expenses across the coverage of an insurance policy, among other areas [12–14].
% 2. (Motivating TV allocation) (i) the principle takes into account dependence between risks, (ii) it defines allocations that can be computed easily, and, (iii) allocation outcomes may effortlessly be explained by managers.

Risk measure, which maps a random loss to a real number, is a common tool to determine the capital reserve for financial institutions. One of the regulatory risk measures used in the realms of banking and insurance is the conditional tail expectation (CTE); see, e.g., \cite{M15}. The CTE satisfies the so-called coherence properties that a desirable risk measure should fulfil (\cite{A99} and \cite{D01}). Consequently, the CTE-based capital allocation can effectively capture the diversification benefits in a portfolio, making it the most important case of the Euler allocation principle (e.g., \cite{D01}, \cite{T04}, and \cite{T08}). Moreover, the CTE-based capital allocation arises as a special case of the optimisation approach to capital allocation as shown in, e.g., \cite{LG04} and \cite{D12}.

Despite the various advantages of the CTE and its allocation method, it has been pointed out that the CTE cannot capture sufficient tail behaviour of the loss distribution. Under severely unfavourable conditions, the actual loss may far exceed the agent's capital reserves based on the CTE. Therefore, this has led to suggestions to supplement the CTE with higher order moments for a more comprehensive evaluation of a portfolio's risk characteristics. In the finance literature, higher moments, most notably skewness and kurtosis, are commonly used in risk assessment; see, e.g., \cite{S70}, \cite{S01}, \cite{H10}, and \cite{AP18}. In the context of capital allocation, the most prominent consideration is the tail variance (TV) (see, e.g., \cite{V04} and \cite{FL06}). However, to our best knowledge, research on capital allocation with the TV is scarce, and no studies have yet considered capital allocation with other tail moments of higher order.
% More broadly, only a few studies explicitly focus on the tail central moments (see \cite{T01}, \cite{LV16}, \cite{EK21} and \cite{Y25}). - these only find TM or TCM for specific distributions, not really applying it to anywhere.

To address this gap in the literature, we first introduce a new capital allocation method based on the tail central moments (TCM), generalising the tail covariance-based capital allocation of \cite{V04} and \cite{FL06}. Secondly, we derive recursive analytical expressions of the TCM and the TCM-based capital allocation for the general class of multivariate normal mean-variance mixture (NMVM) distributions (Theorems \ref{thm:1} and \ref{thm:2}). The NMVM class is known to be extremely flexible and contains many notable members \citep{M15}. One such example is the generalised hyperbolic (GH) distribution, which itself includes the normal, skewed student-$t$, variance Gamma, normal inverse Gaussian, hyperbolic, and other renowned distributions as special cases. The GH distribution is well recognised for its effectiveness in modelling financial and actuarial data due to its connections with the L\'evy process, especially one that exhibits tail behaviour and asymmetry (see, e.g., \cite{EK95}, \cite{N09}, and \cite{KW14}).

This paper contributes to the rich literature of capital allocation for multivariate distributions. The literature on the CTE-based capital allocation is extensive and well-developed. \cite{P02} derived the CTE-based capital allocation for the multivariate normal distribution. This result was later expanded in different directions. One direction considers distributions with heavy tails, such as the elliptical distribution and its extensions (\cite{LV03}, \cite{IL21}, and \cite{IL25}), the GH distribution (\cite{IL15} and \cite{IL19}), and the NMVM class \citep{KK19}. Other directions focus on skewed distributions and compound distributions, see, e.g., \cite{V06} for the CTE-based capital allocation of skewed distributions and \cite{FL10} and \cite{D20} for that of compound distributions. On the contrary, only a few studies have examined the TV-based capital allocation, such as \cite{V04} for the normal distribution, \cite{V05} and \cite{FL06} for the elliptical distribution, \cite{L13} for the lognormal distribution, and e.g., \cite{W14} and \cite{R22} for other applications. Our results broadly contribute to the literature by introducing a novel TCM-based capital allocation to enhance the accuracy of risk assessment. In particular, our results complement those of \cite{KK19}.
% More references on other distributions:
% some included effects of compounding (e.g. the Tweedie distribution by \cite{FL10}, the Pareto distribution by \cite{CL07}, and \cite{C12}, \cite{D20}, \cite{R22} for various compound distributions). 
% the exponential dispersion models (\cite{LV05}), the exponential family class (\cite{K10})
% The topic of capital allocation is of significant interest for the family of elliptical distributions, particularly when higher order risk measures are considered for certain members of this family. \cite{LV03} first considered the CTE for elliptical distributions, and the tail (central) moments of Generalised Skew-Elliptical (GSE) distributions are studied in \cite{EK21}. Subsequently, \cite{L16} and \cite{L18} studied the Multivariate CTE (MTCE) and Multivariate Tail Covariance (MTCov) of elliptical distributions, respectively, which are both extended by \cite{ZY22} and \cite{Z24}, respectively, to GSE distributions. On the other hand, \cite{IL21} and \cite{IL25} studied the CTE and TV of the Generalised Hyper-Elliptical (GHE), respectively.
% For optimal capital allocations based on Tail Mean-Variance measures, see \cite{XM13} and \cite{Y25b}.

The remainder of this paper is organised as follows. Section \ref{sec:2}  introduces the TCM-based capital allocation method and the NMVM class. In Section \ref{sec:3}, recursive analytical expressions for the TCM of the univariate NMVM distribution are derived. Section \ref{sec:4} applies the TCM-based capital allocation to the multivariate NMVM class to obtain explicit expressions for the capital allocated to each component. Section \ref{sec:5} illustrates our theoretical findings with a numerical example based on the multivariate GH distribution. Section \ref{sec:6} concludes.

% Possible areas to include:
% 1. History of NMVM, GHD \dots\dots 
% 2. Use of NMVM, GHD to model financial losses (and connection to Levy process)
% 3. Connection of TCM, TM to Mean-Deviation Measures
% 4. Analytical/Numerical results of finding TM or TCM
% 5. Mentioning of ERM being top-down (our approach is to find TM(S), then use it to find allocation)

\subsection*{Notation}

Denote by $\N_0$ (resp.~$\N$ and $\R_+$) the set of non-negative integers (resp.~positive integers and non-negative real numbers). All vectors are column vectors. For a random variable $X$, we denote by $f_X, F_X, \Fbar_X$, and $h_X$ its density, cumulative distribution, survival and hazard functions, respectively (with $h_X(x)=f_X(x)/\Fbar_X(x)$ for $x\in\R$). For $\al \in (0,1)$, the quantile of  a random variable $X$  is denoted by $x_\al := \inf \cB{x \in \R : \P(X \le x) \ge \al}$. Whenever we consider the $k$-th moment of a random variable $X$, we assume that $\E[|X|^k] < \infty$, where $k\in \N$.
% Throughout this paper, we assume that all distribution functions are absolutely continuous, including the joint distribution of multivariate random vectors. Whenever we consider the $k$-th moment of a random variable $X$, we assume that $\E[|X|^k] < \infty$.

\section{Preliminaries} \label{sec:2}

In this section, we review the definitions of tail moments, capital allocation methods, and the multivariate normal mean-variance mixture distribution. In particular, we introduce a capital allocation method based on the tail central moments. 

% In this section, we list some preliminaries on moments and probability distributions. First, we state our definition of tail moments and tail central moments, and provide some comments about their usage in existing literature. Then, we introduce a capital allocation method based on TCM as the total capital reserve, and compare it to some well-known approaches of allocating capital. Finally, we state the definition of the multivariate normal mean-variance mixture and its related distributions, namely the generalised hyperbolic and generalised inverse-Gaussian distribution. Some relevant properties of these distributions are stated.

\subsection{Tail moments and tail central moments}

The tail moments (TM) and tail central moments (TCM), especially of orders 1 or 2, are commonly used in the literature of capital allocation (see, e.g., \cite{O00}, \cite{V04}, and \cite{KK19}).  

\begin{definition} \label{def:1}
Fix $k \in \N$ and $\al \in (0,1)$. For a random variable $X$, the $k$-th order tail moment (TM) at confidence level $\al$ is defined as
\begin{align*}
\bTMa{k}{X} := \E \sB{X^k \mid X > x_\al}.
\end{align*}
When $k=1$, the TM is referred to as the conditional tail expectation (CTE), denoted by $\aCTEa{X}$.
\end{definition}

\begin{definition} \label{def:2}
Fix $k \in \N$ and $\al\in(0,1)$. For a random variable $X$, the $k$-th order tail central moment (TCM) at confidence level $\al$ is defined as
\begin{align*}
\bTCMa{k}{X} :=
\E \sB{\rB{X - \aCTEa{X}}^k \mid X > x_\al} . 
\end{align*}
When $k=2$, the TCM is referred to as the tail variance (TV).
\end{definition}

\begin{remark} \label{rem:1}
There has been some inconsistency regarding the terminologies of the TM and TCM. The TM and TCM have been referred to as the Tail Conditional Moment in the literature (see, e.g., \cite{K10} and \cite{H19} for the TM and \cite{EK21} for the TCM). When considering an aggregate risk $S=X_1 + \dots + X_n$, \cite{L23} and \cite{Y25} define $\bCES{X^k_i}{S}$ and $\bCES{(X_i-\aCTEa{X_i})^k}{S}$ as the TM and TCM instead.
% In this case, $k=1$ corresponds to what is more commonly known as Marginal Expected Shortfall.
\end{remark}

% \begin{remark} \label{rem:2}
% Other higher order (tail) risk measures have received interest within and beyond the capital allocation literature. Similarly defined examples include (for some $p \geq 1, b \geq 0, t \in \R$, and $x_+=\max(0,x)$): higher order expected shortfall, $\rho(S) = \bCES{(S-s_\al)^p}{S}$ (\cite{T01}); (lower) partial moments, $\rho(S) = \E[(S-t)_+]^p$ (\cite{F77}); coherent risk measures of semi-$\mathcal{L}^p$ type, $\rho(S) = \E[S] + b \E \sB{|(S-\E[S])_+|^p}^{\frac{1}{p}}$, with $\mathcal{L}^p$ being a linear probability space (\cite{R06}). 
% \end{remark}

\begin{remark} \label{rem:2}
Another approach to generalising the CTE is via stochastic optimisation formulas, often with desirable properties preserved. For instance, \cite{K07} and \cite{G22} considered $\rho(X) = \inf_{x \in \R} \cB{x + (1-q)^{-1} \phi(\max(X-x,0))}$, with $\phi(X)=\E \sB{|X|^p}^{1/p}$ and for some $p \geq 1$, $q \in (0,1)$, which is named as the higher moment risk measure. When $p=1$ and $F_X$ is differentiable, we recover the CTE representation in \cite{RU00}.
% \cite{P24} considered the higher order spectral risk $\phi(S)=R_{\sg}(S)$, where $R_{\sg}$ is a risk functional and $\sg(\cdot)$ is a spectral function (see (3.1) in \cite{P24} for details).
\end{remark}

\subsection{Tail central moment-based capital allocation} \label{sec:2.2}
In practice, financial institutions are usually exposed to a portfolio of losses rather than a single loss. The portfolio may consist of policyholders, business lines, or investment assets, depending on the nature of the financial institution. Throughout this paper, we consider an agent with $n \in \N$ random losses $X_1, \dots, X_n$ and denote by $S=X_1+\dots+X_n$ its aggregate loss. After determining the total capital reserve of the aggregate loss $S$, a common practice is to allocate the risk capital to individual losses. Let  $K\in\R$ be the total capital reserve for $S$, and $K_i\in\R$ be the capital allocated to $X_i$ for $i=1,\dots,n$. A capital allocation method is said to satisfy the full allocation property if 
\begin{align*}
K= \iSum K_i.
\end{align*}

One popular capital allocation method is the CTE-based capital allocation, which specifies that 
\begin{align*}
K=\aCTEa{S} \mbox{~~and~~} K_i=\bCES{X_i}{S} \mbox{~~for all~~} i=1, \dots, n.
\end{align*}
It is easy to see that it fulfils the full allocation property. As a coherent allocation principle (see \cite{D01}) with a simple expression, it has received much interest since its introduction in, e.g., \cite{O00}. Nonetheless, the CTE-based capital allocation has certain limitations. In particular, the CTE alone is insufficient in capturing the tail behaviour of losses (e.g., dispersion), which can be crucial to risk management. To address these concerns, we introduce a new class of TCM-based capital allocation methods and discuss some of its properties.

\begin{definition} \label{def:3}
For $k \in \N \setminus \cB{1}$, the $k$-th order TCM-based capital allocation with confidence level $\al \in (0,1)$ is defined as
\begin{align} \label{eq:1}
K = \bTCMa{k}{S} \mbox{~~and~~}
K_i = \aTCOVS{X_i, \rB{S-\aCTEa{S}}^{k-1}} \mbox{~~for all~~} i=1, \dots, n.
\end{align}
\end{definition}

The TCM-based capital allocation provides direct interpretations of the risk contributions of individual losses to the aggregate loss. For instance, if $k=2$, the TCM-based capital allocation method recovers the TV-based capital allocation\footnote{It is referred to as the tail covariance-based capital allocation in \cite{V04} and \cite{FL06}.} in \cite{V04} and \cite{FL06}, i.e., 
\begin{align*}
K_i=\aTCOVS{X_i, S} \mbox{~~for all~~} i=1, \dots, n.
\end{align*}
The TV-based capital allocation thus quantifies the dependence between individual losses and the aggregate loss in tail regions. The TCM-based capital allocation can also capture relationships between the aggregate tail dispersion and each component. One example is $k=3$, with 
\begin{align*}
K_i = \aTCOVS{X_i, \rB{S-\aCTEa{S}}^2} \mbox{~~for all~~} i=1, \dots, n.
\end{align*}
Note that the TCM-based capital allocation can be negative, which shows a diversification benefit.

\begin{proposition} \label{prop:1}
The TCM-based capital allocation satisfies the full allocation property. 
\end{proposition}
\begin{proof} 
Let $S_\al = S - \mathrm{CTE}_\al(S)$. We have
\begin{align*}
\iSum K_i
=&~ \iSum \aTCOVS{X_i, \rB{S-\aCTEa{S}}^{k-1}} \\
=&~ \iSum \rB{ \bCES{X_i S_\al^{k-1}}{S}
    - \bCES{X_i}{S} \bCES{S_\al^{k-1}}{S} } \\
=&~ \bCES{\iSum X_iS_\al^{k-1}}{S}
    - \bCES{S_\al^{k-1}}{S} \iSum \bCES{X_i}{S} \\
=&~ \cTCMb{S}{S}{k-1} - \cTCMb{\aCTEa{S}}{S}{k-1} \\
=&~ \aTCMS{k}=K. \qedhere
\end{align*}
\end{proof}

As the CTE alone does not adequately characterise the tail behaviour of losses,  it is worth considering linear combinations of the CTE-based and TCM-based capital allocation methods. For instance, an overall capital reserve of 
\begin{align} \label{eq:2}
K=m_1 \aCTEa{S} + m_2 \mathrm{TV}_\al(S) + m_3 \bTCMa{3}{S}, 
\end{align} 
for some $m_1, m_2, m_3 \in \R_+$, not only measures the average tail loss, but also takes into account other characteristics of the tail region such as dispersion and asymmetry. The corresponding capital allocation is feasible due to linearity. The combination allows a lot of flexibility to the agent when deciding their portfolio management priorities. The idea of combining the CTE and TV has been considered by, e.g., \cite{FL06}, \cite{IL15}, and \cite{KK19} as a premium principle for the entire portfolio, with only \cite{FL06} applying it to capital allocation. We extend this idea by including the 3rd order TCM as well, and demonstrate it via a real-data analysis in Section \ref{sec:5}.

\begin{remark} \label{rem:3}
The Euler allocation principle is a popular capital allocation method. This is because it possesses the full allocation property as well as other desirable properties, and it aligns with concepts from other disciplines such as economics and game theory (see Section 2.2 of \cite{T08} and references therein for detailed discussions). While the TCM does not fulfil the conditions for the Euler allocation principle, we can modify it by ``rooting'' the TCM so that the Euler allocation principle can be applied under mild assumptions of the random losses, with the following allocation outcome:
\begin{align*} 
K = \bTCMa{k}{S}^{\frac{1}{k}} \mbox{~~and~~}
K_i = \frac{ \aTCOVS{X_i, (S-\aCTEa{S})^{k-1} } }{
    \bTCMa{k}{S}^{1-\frac{1}{k}} } \mbox{~~for all~~} i=1, \dots, n;
\end{align*}
see Appendix \ref{app:A} for the derivation of this result. Clearly, switching between \eqref{eq:1} and the above has no additional computational difficulty. Moving forward, \eqref{eq:1} in Definition \ref{def:3} will be used for its neater expressions. The case when $k=2$, together with the CTE-based capital allocation, is studied in \cite{FL06} and \cite{G21} as the risk-adjusted tail value-at-risk allocation method. 
\end{remark}

\subsection{Normal mean-variance mixture distributions}

The following definition follows Definition 3.11 of \cite{M15}.

\begin{definition} \label{def:4}
A random vector $\mathbf{X}$ is said to follow an $n$-dimensional normal mean-variance mixture (NMVM) distribution if
\begin{align*}
\mathbf{X} \overset{d}{=} \mathbf{m}(\Th) + \sqrt{\Th}A\mathbf{Z},
\end{align*}
where
 \begin{enumerate}[(i)]
  \item $\mathbf{Z} \sim MVN_k(\mathbf{0}, \mathbf{I_k})$ is a $k$-dimensional standard multivariate normal random vector with the identity variance-covariance matrix;
  \item $A \in \R^{n \times k}$ is a matrix;
  \item $\Th$ is a non-negative random variable, independent of $\mathbf{Z}$, with density function $\pi(\th)$ for $\th>0$. It is referred to as the mixing random variable; 
  \item $\mathbf{m} : [0, \infty) \rightarrow \R^n$ is a measurable function of $\Th$. 
\end{enumerate}
\end{definition}
Throughout this paper, we assume that $\mathbf{m}(\Th) = \bm{\mu} + \Th \bm{\gm}$ where $\bm{\mu}, \bm{\gm} \in\R^n$. Let $\Sigma := AA' = (\SGij)_{1 \le i,j \le n}$. We will specify an NMVM random variable (or its distribution) via the parameters $\bm{\mu}, \bm{\gm}$, and $\Sigma$, and the mixing random variable $\Theta$. For a univariate NMVM random variable, we write the parameters as $\mu:=\bm{\mu}$, $\gamma:=\bm{\gm}$, and $\sg^2:=\Sigma$. 
 
We present below some useful properties of the NMVM distribution. First, it is clear that
\begin{align*}
\mathbf{X} \mid \Th = \th \sim  MVN_n\rB{\mathbf{m}(\th), \th \Sigma)}.
\end{align*}
Second, the class of NMVM distributions is closed under linear transformations (see, e.g., Proposition 2.1 of \cite{KK19}). This is a useful property with many financial applications, such as when portfolio weights are concerned. In particular, it follows that $S=X_1+\dots+X_n$ is an NMVM random variable with mixing random variable $\Th$ and parameters $\mu= \mathbf{1'} \bm{\mu}, \sg^2= \mathbf{1' \Sigma 1},$ and $\gm=\mathbf{1'} \bm{\gm}$. In general, NMVM distributions are not elliptical, and $\bm{\mu}$ and $\Sigma$ are not the mean vector and covariance matrix of $\mathbf{X}$.

The NMVM class contains many important distributions, one of which is the generalised hyperbolic (GH) distribution, where $\Th$ follows a generalised inverse Gaussian (GIG) distribution with three parameters $\lm\in\R$ and $\chi,\psi\ge 0$. We denote a $n$-dimensional multivariate GH distribution by $MGH_n(\lm,\ \chi,\ \psi,\ \bm{\mu},\ \bm{\Sigma},\ \bm{\gamma})$. The density of the GIG distribution is given by
\begin{align*}
\PIt = \frac{\chi^{-\lm} (\sqrt{\chi \psi})^\lm}
            {2 \mathcal{K}_{\lm}(\chi \psi) } 
    \th^{\lm-1} \exp \rB{-\frac{1}{2}(\chi \th^{-1} + \psi \th)}, 
  ~~ \th>0 ,
\end{align*}
where $\mathcal{K}_{\lm}$ is a modified Bessel function of the second kind with index $\lm$:
\begin{align*}
\mathcal{K}_{\lm}(z) = \frac{1}{2} 
  \INTx{ x^{\lm-1} e^{-\frac{1}{2} z (x^{-1} + x)} }.
\end{align*}
The parameters need to satisfy one of: $\chi > 0, \psi \ge 0$ when $\lm < 0$; $\chi > 0, \psi > 0$ when $\lm = 0$; $\chi \ge 0, \psi > 0$ when $\lm > 0$. The GIG distribution itself contains the Gamma and inverse Gamma as special cases, and the GH class has several notable members, as listed in the introduction. For more information about the GIG and GH distributions, refer to \cite{J82} or Section 6.2.3 of \cite{M15}.

\section{Tail moments of univariate NMVM distributions} \label{sec:3}

In this section, we derive an analytical solution to the TCM of the aggregate loss $S$ faced by the agent as outlined at the start of Section \ref{sec:2.2}. The model setup and assumptions for Sections \ref{sec:3} and \ref{sec:4} are as follows:
\begin{enumerate} [(i)]
  \item The losses $X_1, \dots, X_n$ follow the multivariate NMVM distribution  (Definition \ref{def:4});
  \item As the NMVM model is closed under linear combinations, the aggregate loss $S$ follows a univariate NMVM distribution with mixing random variable $\Th$ and parameters $\mu= \mathbf{1'} \bm{\mu}$, $\sg^2= \mathbf{1' \Sigma 1}$, and $\gm = \mathbf{1'} \bm{\gm}$;
  \item Fix $l \in \N_0$. We assume that there always exists some $\Cl > 0$ such that $\pi^{*(l)}(\th) := (\Cl)^{-1} \th^l \PIt$ is a valid density function. Let $\Ca$ and $\PIta$ (resp.~$\Cb$ and $\PItb$) be the shorthand notation of $\Cl$ and $\pi^{*(l)}(\th)$ for $l=1$ (resp.~$l=2$).
  \item Denote by $\SL$ an NMVM random variable with the same parameters as $S$, except that the density of its mixing random variable is $\pi^{*(l)}(\th)$. Define $\ALl=1-\Fbar_{\SL}(s_\al)$ for some $\al \in (0,1)$. Let $S^*$ and $\ALa$ (resp.~$S^{**}$ and $\ALb$) be the shorthand notation of $\SL$ and $\ALl$ for $l=1$ (resp.~$l=2$). 
\end{enumerate}

Based on (ii) above, the task in this section reduces to finding the TCM of a univariate NMVM distribution. The solution is achieved through a recursive approach. As a necessary step in calculating the TCM, we also provide recursive formulas for the TM. As a direct consequence, we obtain an explicit formula for the 2nd order TM and TCM of $S$, studied by \cite{KK19}, using different techniques.
 
We first provide the following results, which will be useful in the derivation of Theorem \ref{thm:1}. 

\begin{lemma} {\rm \citep[][Example 3.1]{LV16}} \label{lmm:1}
Fix $k \in \N$, $\mu \in \R$, and $c,\ \sg \in \R_+$. For a random variable $X \sim N(\mu, \sg^2)$, the $k$-th order TM of $X$ follows the recursive relationship
\begin{align}
\E \sB{X^k \mid X>c}
= \sg^2 c^{k-1} \frac{f_X(c)}{\Fbar_X(c)} 
  + \mu \E \sB{X^{k-1} \mid X>c} + (k-1) \sg^2 \E \sB{X^{k-2} \mid X>c} .
\end{align}
\end{lemma}

\begin{lemma} \label{lmm:2}
For some fixed $k \in \N$, $l \in \N_0$, and $\al \in (0,1)$, we have 
\begin{align*}
\bCES{(\SL)^k}{\SL}
= \frac{1}{1-\ALl} \INTth{ 
    \Fbar_{\SL | \th}(s_\al) \bCESt{(\SL)^k}{\SL} \pi^{*(l)}(\th) } .
\end{align*}
\end{lemma}

\begin{proof}
Let random variable $\Th^{*(l)}$ have density $\pi^{*(l)}(\th)$, with $\th>0$. We have 
\begin{align*}
\bCES{(\SL)^k}{\SL}
=&~ \frac{1}{1-\ALl} \aINTSa{ s^k f_{\SL}(s) } \\
=&~ \frac{1}{1-\ALl} \aINTSa{  s^k \INTth{ f_{\SL , \Th^{*(l)}}(s,\theta)  } } \\
=&~ \frac{1}{1-\ALl} \aINTSa{ \INTth{ 
    s^k f_{\SL | \th}(s) \pi^{*(l)}(\th) } } \\
=&~ \frac{1}{1-\ALl} \INTth{ \Fbar_{\SL | \th}(s_\al) \rB{
    \frac{1}{\Fbar_{\SL | \th}(s_\al)} \aINTSa{ 
      s^k f_{\SL | \th}(s) } } \pi^{*(l)}(\th) } \\
=&~ \frac{1}{1-\ALl} \INTth{ 
    \Fbar_{\SL | \th}(s_\al) \E \sB{(\SL)^k \mid \SL>s_\al, \Th^{*(l)}} = \th} \pi^{*(l)}(\th) . \qedhere
\end{align*}
\end{proof}

Now we state our main result for the TM and TCM of the NMVM random variable $S$. 

\begin{theorem} \label{thm:1} 
For $k\in\N$, the $k$-th order TM and TCM of the NMVM random variable $S$ at confidence level $\al\in(0,1)$ can be found recursively by
\begin{align} \label{eq:4} 
\bCES{S^k}{S} 
=&~ \mu \bCES{S^{k-1}}{S} + \FRc \sg^2 s_\al^{k-1} h_{S^*}(s_\al)
\nonumber \\ &~~
  + \FRc \rB{ 
    \gm \bCES{(S^*)^{k-1}}{S^*}
    + (k-1) \sg^2 \bCES{(S^*)^{k-2}}{S^*} 
  } , 
\end{align}
where 
\begin{align} \label{eq:5}
\bCES{(\SL)^k}{\SL} 
=&~ \mu \bCES{(\SL)^{k-1}}{\SL} + \FRl \sg^2 s_\al^{k-1} h_{\SLl}(s_\al) 
\nonumber \\ &~~ 
  + \FRl \Big( \gm \bCES{(\SLl)^{k-1}}{\SLl}
\nonumber \\ &~~~~~~~~~~~~~~~~~~~~~~~~~~~~
    + (k-1) \sg^2 \bCES{(\SLl)^{k-2}}{\SLl} \Big) , 
\end{align}
with \eqref{eq:4} being a special case of \eqref{eq:5} with $l=0$, and 
\begin{align} \label{eq:6}
\bTCMa{k}{S} = \sum^k_{j=0} \binom{k}{j} \bCES{S^{k-j}}{S} (-\aCTEa{S})^j . 
\end{align}
\end{theorem}

\begin{proof}
We will first prove \eqref{eq:4}. We begin with applying Lemma \ref{lmm:1} to obtain
\begin{align*}
\bCESt{S^k}{S} 
=&~ \th \sg^2 s_\al^{k-1} \frac{f_{\ST}(s_\al)}{\Fbar_{\ST}(s_\al)}
  + (\mu + \th \gm) \bCESt{S^{k-1}}{S} \\
  &~~ + (k-1) \th \sg^2 \bCESt{S^{k-2}}{S} .  
\end{align*}
Then, applying the above result and Lemma \ref{lmm:2} (with $l=0$) gives
\begin{align*}
\bCES{S^k}{S} 
=&~ \FRa \INTth{\bCESt{S^k}{S} \FBARSt \PIt} \\
=&~ \FRa \INTth{
    \mu \bCESt{S^{k-1}}{S} \FBARSt \PIt
    + \sg^2 s_\al^{k-1} f_{\St}(s_\al) (\th \PIt)
\\ &~~~~~~~~~~~~~~~~
    + \gm \bCESt{S^{k-1}}{S} \FBARSt (\th \PIt)
\\ &~~~~~~~~~~~~~~~~
    + (k-1) \sg^2 \bCESt{S^{k-2}}{S} \FBARSt (\th \PIt)
  } \\
=&~ \mu \bCES{S^{k-1}}{S} 
  + \Ca \FRa \sg^2 s_\al^{k-1} f_{S^*}(s_\al) \cdot \frac{\Fbar_{S^*}(s_\al)}{\Fbar_{S^*}(s_\al)}
\\ &~~
  + \FRc \gm \bCES{(S^*)^{k-1}}{S^*}
  + \FRc (k-1) \sg^2 \bCES{(S^*)^{k-2}}{S^*} \\
=&~ \FRc \sg^2 s_\al^{k-1} h_{S^*}(s_\al) 
  + \mu \bCES{S^{k-1}}{S} 
\\ &~~
  + \FRc \gm \bCES{(S^*)^{k-1}}{S^*}
  + \FRc (k-1) \sg^2 \bCES{(S^*)^{k-2}}{S^*} .
\end{align*}
Equation \eqref{eq:5} is proven in the same way, as $\SL$ and $S$ are both NMVM random variables, and $(\SL)^*=\SLl$ by definition. Since $\PIt$ is an arbitrary density function, we can replace $\PIt$ with $\pi^{*(l)}(\th)$. Consequently, $S$ (resp.~$S^*$) is replaced with $\SL$ (resp.~$\SLl$), and the rest follows.

Lastly, directly applying binomial expansion onto the TCM of $S$ completes the proof for \eqref{eq:6}.
\end{proof}

\begin{remark}
If we further assume that $\Th \sim GIG(\lm, \chi, \psi)$, then $S \sim GH(\lm, \chi, \psi, \mu, \sg, \gm)$ (see Definition \ref{def:4}). This gives $S^* \sim GH(\lm+1, \chi, \psi, \mu, \sg, \gm)$ ((25) to (27) of \cite{KK19}). This is useful for Section \ref{sec:5}, where recursive formulas for the TM of GH distributed random variables are computed.
\end{remark}

The following corollary presents a particularly interesting case of Theorem \ref{thm:1} when orders of moment are 1 and 2; these results were first obtained by \cite{KK19} (see their Theorem 3.1, Proposition 5.1, and Theorem 5.2).

\begin{corollary} \label{cor:1}
The CTE of the NMVM random variable $S$ at confidence level $\alpha\in(0,1)$ is given by 
\begin{align} \label{eq:7}
\aCTEa{S} 
= \mu + \Ca \rB{\FRb} \rB{\gm + \sg^2 h_{S^*}(s_\al)} ,
\end{align}
and the $2$-nd order TM and TCM of $S$ are given by 
\begin{align} \label{eq:8}
\bTMa{2}{S} 
=&~ \mu^2 + \FRc \rB{\sg^2 + 2\mu \gm + \sg^2(s_\al + \mu) \aHS{*}}
\nonumber \\ &~~
  + \FRd \rB{\gm^2 + \gm \sg^2 \aHS{**}},  
\end{align}
and 
\begin{align} \label{eq:9} 
\mathrm{TV}_\al(S)
=&~ \FRc \sg^2 \rB{1 + (s_\al - \mu) \aHS{*}}
  + \FRd \gm \rB{\gm + \sg^2 \aHS{**}} \nonumber \\
  &~~ - \rB{\FRc (\gm + \sg^2 \aHS{*})}^2 . 
\end{align}
\end{corollary}

\begin{proof}
Equation \eqref{eq:7} is directly obtained from substituting $k=1$ into \eqref{eq:4} in Theorem \ref{thm:1}. Substituting $l=1$, $k=1$ into \eqref{eq:5} in Theorem \ref{thm:1} gives
\begin{align*}
\bCES{S^*}{S^*}
= \mu + \frac{\Cb (1-\ALb)}{\Ca (1-\ALa)} (\gm + \sg^2 \aHS{**}) .
\end{align*}
Applying Theorem \ref{thm:1}, then substituting \eqref{eq:7} and the above result gives
\begin{align*}
\bTMa{2}{S}
=&~ \FRc \sg^2 s_\al \aHS{*} + \mu \bCES{S}{S} \\
  &~~ + \FRc \gm \bCES{S^*}{S^*} + \FRc \sg^2 \\
=&~ \FRc \sg^2 s_\al \aHS{*} 
  + \mu \rB{\mu + \FRc (\gm + \sg^2 \aHS{*})} 
\\ &~~
  + \FRc \gm \rB{ 
    \mu + \frac{\Cb (1-\ALb)}{\Ca (1-\ALa)} (\gm + \sg^2 \aHS{**}) }
  + \FRc \sg^2 ,
\end{align*}
and we obtain \eqref{eq:8} after routine algebraic simplification. Subsequently, \eqref{eq:9} is obtained by
\begin{align*}
\mathrm{TV}_\al(S)
=&~ \bCES{S^2}{S} - \bCES{S}{S}^2 \\
=&~ \mu^2 + \FRc \rB{\sg^2 + 2\mu \gm + \sg^2(s_\al + \mu) \aHS{*}}
\\ &~~
  + \FRd \rB{\gm^2 + \gm \sg^2 \aHS{**}} 
  - \rB{\mu + \FRc (\gm + \sg^2 \aHS{*})}^2 \\
=&~ \mu^2 - \mu^2 
  + \FRc \rB{\sg^2 + 2 \mu \gm - 2 \mu \gm + \sg^2 (s_\al + \mu - 2 \mu) \aHS{*}}
\\ &~~
  + \FRd \rB{\gm^2 + \gm \sg^2 \aHS{**}} - \rB{\FRc (\gm + \sg^2 \aHS{*})}^2 \\
=&~ \FRc \sg^2 \rB{1 + (s_\al - \mu) \aHS{*}}
\\ &~~
  + \FRd \gm \rB{\gm + \sg^2 \aHS{**}} - \rB{\FRc (\gm + \sg^2 \aHS{*})}^2 . \qedhere
\end{align*}
\end{proof}

\section{Capital allocation for multivariate NMVM distributions}  \label{sec:4}

In Section \ref{sec:3},  the TCM of the aggregate loss $S$ has been derived. Next, we proceed to study the TCM-based capital allocation method as defined in Definition \ref{def:3}. Again, we obtain an explicit formula for the 2nd-order TCM-based capital allocation as a special case.

Recall the same model setup as in Section \ref{sec:3}. In addition, let $\SGS := \iSum \SUMj \SGij$ and $\SGiS := \SUMj \SGij$ for $i \in \cB{1,\dots,n}$. We also denote the NMVM random vector parameters by $\bm{\mu} = (\mu_1, \dots, \mu_n)' \in \R^n$, $\bm{\gm} = (\gm_1, \dots, \gm_n)' \in \R^n$, and $\bm{\Sigma} = (\SGij)_{1 \le i,j \le n} \in \R^{n \times n}$. We start by stating some useful results.

\begin{lemma} {\rm \citep[][Theorem 4.1]{KK19}} \label{lmm:3}
Consider the multivariate NMVM random vector $(X_1, \dots, X_n)$ with mixing random variable $\Th$ and parameters $\bm{\mu}$, $\bm{\gm}$, and $\bm{\Sigma}$, and the aggregate loss $S$. Under the CTE-based capital allocation with confidence level $\al \in (0,1)$, the capital allocated to $X_i$ for all $i = 1, \dots, n$, is given by
\begin{align*}
K_i = \bCES{X_i}{S} = \AIa + \AIb \bCES{S}{S} + \AIc \FRc ,
\end{align*}
where the coefficients $\AIa$, $\AIb$, and $\AIc$ are defined as
\begin{align*}
\AIa = \MUi - \AIb \SUMj \MUi ,~~
\AIb = \frac{\SUMj \SGij}{\SGS} ,~~\mbox{and}~~
\AIc = \GMi - \AIb \SUMj \GMj .
\end{align*}
\end{lemma}

\begin{lemma} \label{lmm:4}
Consider the same random variables $X_1, \dots, X_n$, and $S$ in Lemma \ref{lmm:3}, as well as all related parameters and coefficients. The random vector $(X_1, \dots, X_n \mid S=s, \Th=\th)$ for some $s \in \R$, $\th \in \R_+$ follows a multivariate normal distribution, with
\begin{align*}
\aCEEst{X_i} = \AIa + \AIb s + \AIc \th \mbox{~~for all~~} i=1, \dots, n, 
\end{align*}
and
\begin{align*}
\aCOVst{X_i, X_j} = \th(\SGij - \AIb \AJb \SGS) \mbox{~~for all~~} i,j=1, \dots, n.
\end{align*}
\end{lemma}

\begin{proof} 
Since the random vector $(X_1, \dots, X_n \mid \Th=\th)$ follows a multivariate normal distribution (see Definition \ref{def:4}), Theorem 3.3.3 of \cite{T12} implies that $(X_1, \dots, X_n, S \mid \Th=\th)$ also follows a multivariate normal distribution.  By Theorem 3.3.4 of \cite{T12}, $(X_1, \dots, X_n \mid S=s, \Th=\th)$ follows a multivariate normal distribution with its mean and covariance given by
\begin{align*}
\aCEEst{X_i} 
=&~ \E[X_i \mid \Th=\th] + \frac{\Cov[X_i, S \mid \Th=\th]}{\Cov[S, S \mid \Th=\th]} (s-\E[S \mid \Th=\th]) \\ 
=&~ (\mu_i + \th \gm_i) - \frac{\SUMj \SGij}{\SGS} \SUMk (\mu_k + \th \gm_k) + \frac{\SUMj \SGij}{\SGS} s \\ 
=&~ \MUi - \AIb \SUMk \MUk + \th \rB{\GMi - \AIb \SUMk \GMk } + \AIb s \\
=&~ \AIa + \AIb s + \AIc \th ,
\end{align*}
and
\begin{align*}
\aCOVst{X_i, X_j} 
=&~ \Cov[X_i, X_j \mid \Th=\th]
  - \frac{\Cov[X_i, S \mid \Th=\th] ~ \Cov[S, X_j \mid \Th=\th]}{\Cov[S, S \mid \Th=\th]} \\
=&~ \th \SGij - \frac{(\th \SGiS) (\th \SGjS)}{\th \SGS} \\
=&~ \th (\SGij - \AIb \AJb \SGS) . \qedhere
\end{align*} 
\end{proof}

Before arriving at the TCM-based capital allocation, we provide a useful intermediate result. 
\begin{proposition} \label{prop:2}
Consider the same random variables $X_1, \dots, X_n$, and $S$ in Lemma \ref{lmm:3}, as well as all related parameters and coefficients. Fix $k \in \N \setminus \cB{1}$ and $\al \in (0,1)$. For all $i \in \cB{1,\dots,n}$, we have 
\begin{align*}
\aTCOVS{X_i, S^{k-1}} 
=&~ \AIb \rB{ \bCES{S^k}{S} - \bCES{S}{S} \bCES{S^{k-1}}{S} }
\\ &~~
  + \AIc \FRc \rB{ \bCES{(S^*)^{k-1}}{S^*} - \bCES{S^{k-1}}{S} } .
\end{align*}
\end{proposition}

\begin{proof}
Using similar techniques to those in the proof of  Lemma \ref{lmm:2}, we obtain
\begin{align*}
\aTCOVS{X_i, S^{k-1}} 
=&~ \bCES{X_i S^{k-1}}{S} - \bCES{X_i}{S} \bCES{S^{k-1}}{S} \\
=&~ \FRa \aINTSa{ s^{k-1} \aCEEs{X_i} \FS } \\
&~~ - \bCES{X_i}{S} \bCES{S^{k-1}}{S}.
\end{align*}
Explicit solutions to $\bCES{X_i}{S}$ and $\bCES{S^{k-1}}{S}$ are available in Lemma \ref{lmm:3} and Theorem \ref{thm:1}, respectively. Equation (45) in \cite{KK19} states that
\begin{align*}
\aCEEs{X_i} = \AIa + \AIb + \AIc \Ca \frac{f_{S^*}(s)}{\FS} .
\end{align*}
Thus, we have
\begin{align*}
\bCES{X_i S^{k-1}}{S} 
=&~ \FRa \aINTSa{ 
    s^{k-1} \rB{\AIa + \AIb s + \AIc \Ca \frac{f_{S^*}(s)}{\FS}} \FS } \\
=&~ \FRa \aINTSa{\AIa s^{k-1} \FS + \AIb s^k \FS + \AIc \Ca s^{k-1} f_{S^*}(s) } \\
=&~ \AIa \bCES{S^{k-1}}{S}  + \AIb \bCES{S^k}{S} 
\\ &~~ + \AIc \FRc \bCES{(S^*)^{k-1}}{S^*} ,
\end{align*}
which gives
\begin{align*}
\aTCOVS{X_i, S^{k-1}} 
=&~ \bCES{X_i S^{k-1}}{S} - \bCES{X_i}{S} \bCES{S^{k-1}}{S} \\
=&~ \AIa \bCES{S^{k-1}}{S}  + \AIb \bCES{S^k}{S} 
\\ &~~ 
  + \AIc \FRc \bCES{(S^*)^{k-1}}{S^*} 
\\ &~~
  - \rB{\AIa + \AIb \bCES{S}{S} + \AIc \FRc} \bCES{S^{k-1}}{S} \\
=&~ \AIb \rB{\bCES{S^k}{S} - \bCES{S}{S} \bCES{S^{k-1}}{S}} 
\\ &~~
  + \AIc \FRc \rB{ \bCES{(S^*)^{k-1}}{S^*} -  \bCES{S^{k-1}}{S} } . \qedhere
\end{align*}
\end{proof}

Now we state our main result in capital allocation.

\begin{theorem} \label{thm:2} 
Consider the same random variables $X_1, \dots, X_n$, and $S$ in Lemma \ref{lmm:3}, as well as all related parameters and coefficients. For some $k \in \N \setminus \cB{1}$, under the $k$-th order TCM-based capital allocation in Definition \ref{def:3} with confidence level $\al \in (0,1)$, the capital allocated to $X_i$ for all $i = 1, \dots, n$, is given by
\begin{align} \label{eq:10}
K_i
=&~ \aTCOVS{X_i, (S-\aCTEa{S})^{k-1}} \nonumber \\ 
=&~ \AIb \bTCMa{k}{S}  
  + \AIc \FRc \rB{ \bCES{ \rB{S^*-\aCTEa{S}}^{k-1} }{S^*} - \bTCMa{k-1}{S} } .
\end{align}
\end{theorem}

\begin{proof}
Using the binomial expansion and Proposition \ref{prop:2}, we have
\begin{align} \label{eq:11}
K_i
=&~ \sum^{k-1}_{j=0} \binom{k-1}{j} (-\aCTEa{S})^j \aTCOVS{X_i, S^{k-1-j}} \nonumber \\
=&~ \sum^{k-2}_{j=0} \binom{k-1}{j} (-\aCTEa{S})^j \Bigg(
    \AIb \rB{\bCES{S^{k-j}}{S} - \bCES{S}{S} \bCES{S^{k-1-j}}{S}} 
\nonumber \\ &~~~~~~~~~~~~~~~~~~~~~~~~~~~~~~~~~~~~~~~
    + \AIc \FRc \Big( \bCES{(S^*)^{k-1-j}}{S^*} 
\nonumber \\ &~~~~~~~~~~~~~~~~~~~~~~~~~~~~~~~~~~~~~~~~~~~~~~~~~~~~~~~~~~~~~
      - \bCES{S^{k-1-j}}{S} \Big)
  \Bigg) + 0 \nonumber \\
=&~ \AIb \sum^{k-2}_{j=0} \binom{k-1}{j} 
    \rB{ (-\aCTEa{S})^j \bCES{S^{k-j}}{S} + (-\aCTEa{S})^{j+1} \bCES{S^{k-1-j}}{S} } 
\nonumber \\ &~~
  + \AIc \FRc \sum^{k-2}_{j=0} \binom{k-1}{j} \Bigg( 
    (-\aCTEa{S})^j \bCES{(S^*)^{k-1-j}}{S^*} 
\nonumber \\ &~~~~~~~~~~~~~~~~~~~~~~~~~~~~~~~~~~~~~~~~~~
    - (-\aCTEa{S})^j \bCES{S^{k-1-j}}{S} \Bigg) .
\end{align}
For the latter summation (with coefficient $\AIc \Ca (1-\ALa)/\rB{1-\al}$), we have 
\begin{align} \label{eq:12}
&~ \sum^{k-2}_{j=0} \binom{k-1}{j} (-\aCTEa{S})^j 
    \rB{ \bCES{(S^*)^{k-1-j}}{S^*} - \bCES{S^{k-1-j}}{S} } \nonumber \\
=&~ \rB{ \sum^{k-1}_{j=0} \binom{k-1}{j} 
    (-\aCTEa{S})^j \bCES{(S^*)^{k-1-j}}{S^*}
  - (-\aCTEa{S})^{k-1} }
\nonumber \\ &~~
  - \rB{ \sum^{k-1}_{j=0} \binom{k-1}{j} 
    (-\aCTEa{S})^j \bCES{S^{k-1-j}}{S} 
  - (-\aCTEa{S})^{k-1} } \nonumber \\
=&~ \bCES{ \rB{S^*-\aCTEa{S}}^{k-1} }{S^*} - \aTCMS{k-1} \nonumber \\
=&~ \bCES{ \rB{S^*-\aCTEa{S}}^{k-1} }{S^*} - \bTCMa{k-1}{S} , 
\end{align}
with the second-to-last equality being an application of the binomial theorem. For the former summation (with coefficient $\AIb$), we first notice that
\begin{align*}
&~ (-\aCTEa{S})^{k-1} \bCES{S^{k-(k-1)}}{S} + (-\aCTEa{S})^{k-1+1} \bCES{S^{k-1-(k-1)}}{S} \\
=&~ (-\aCTEa{S})^{k-1} \aCTEa{S} + (-\aCTEa{S})^k 
= 0.
\end{align*}
Using the identity above, we have
\begin{align} \label{eq:13}
&~ \sum^{k-2}_{j=0} \binom{k-1}{j} (-\aCTEa{S})^j \bCES{S^{k-j}}{S} 
  + \sum^{k-2}_{j=0} \binom{k-1}{j} (-\aCTEa{S})^{j+1} \bCES{S^{k-1-j}}{S} \nonumber \\
=&~ \sum^{k-1}_{j=0} \binom{k-1}{j} (-\aCTEa{S})^j \bCES{S^{k-j}}{S} 
  + \sum^{k-1}_{j=0} \binom{k-1}{j} (-\aCTEa{S})^{j+1} \bCES{S^{k-1-j}}{S} \nonumber \\
=&~ \bCES{S^k}{S} 
  + \sum^{k-1}_{j=1} \binom{k-1}{j} (-\aCTEa{S})^j \bCES{S^{k-j}}{S} 
\nonumber \\ &~~
  + \sum^{k}_{j=1} \binom{k-1}{j-1} (-\aCTEa{S})^j \bCES{S^{k-j}}{S} \nonumber \\
=&~ \bCES{S^k}{S} 
  + \sum^{k-1}_{j=1} \binom{k-1}{j} (-\aCTEa{S})^j \bCES{S^{k-j}}{S}
\nonumber \\ &~~
  + \sum^{k-1}_{j=1} \binom{k-1}{j-1} (-\aCTEa{S})^j \bCES{S^{k-j}}{S}
  + (-\aCTEa{S})^k \nonumber \\
=&~ \bCES{S^k}{S} 
  + \sum^{k-1}_{j=1} \binom{k}{j} (-\aCTEa{S})^j \bCES{S^{k-j}}{S} 
  +  (-\aCTEa{S})^k \nonumber \\
=&~ \sum^k_{j=0} \binom{k}{j} (-\aCTEa{S})^j \bCES{S^{k-j}}{S} 
= \bTCMa{k}{S} ,
\end{align}
where the binomial theorem is used at the last equality, and the identity $\binom{k-1}{j} + \binom{k-1}{j-1} = \binom{k}{j}$ is used at the fourth equality. Finally, \eqref{eq:10} is obtained by substituting \eqref{eq:12} and \eqref{eq:13} into \eqref{eq:11}.
\end{proof}

The capital allocation expressions in Theorem \ref{thm:2} can be seen as the sum of two components, signified by the terms with coefficients $\AIb$ and $\AIc$ in \eqref{eq:10}, which are the only variables that are specific to each loss $X_i$. Based on the representations of $\AIb$ and $\AIc$ in Lemma \ref{lmm:3}, the variable $\AIb$ represents a direct risk contribution from $X_i$ to the total risk $\bTCMa{k}{S}$, whereas $\AIc$ shows the indirect adjustments required to reflect other tail behaviours such as tail skewness. The existence of these interpretations allows agents to explain their capital allocation outcome to other stakeholders more easily, while maintaining the rigorous mathematical results that support their complex risk management priorities.
% \com{Yuyu: okay and don't feel it is easy to skip our results and go for simulation which works quite badly for (1) very heavy-tailed risks (2) multivariate case, and (3) tail region. The numerical results are usually unstable and require a huge number of simulations. The simulation part of Kim and Kim I guess is to answer referee's question which you have to answer}
% In addition, agents who prefer numerical methods can still benefit from Theorem \ref{thm:2}, even without Theorem \ref{thm:1}. Instead of simulating tail observations of a multivariate random vector, which is both computationally costly and prone to random fluctuations, it is only required to simulate sufficient observations of $S$ and $S^*$, then apply Theorem \ref{thm:2} to find the TCM-based capital allocation outcome. This greatly reduces the potential dimensionality problem of the simulation task and ensures more consistent allocation outcomes.

As results of the second order will naturally be of more interest for their intuitive interpretation, we provide explicit results of the $2$-nd order TCM-based order capital allocation, which is also known as the TV-based capital allocation (see Definition \ref{def:3}). 

\begin{corollary} \label{cor:2}
Consider the same random variables $X_1, \dots, X_n$, and $S$ in Lemma \ref{lmm:3}, as well as all related parameters and coefficients. Under the TV-based capital allocation with confidence level $\al \in (0,1)$, the capital allocated to $X_i$ for all $i = 1, \dots, n$, is given by
\begin{align*}
K_i
= \aTCOVS{X_i, S} 
= \AIb \mathrm{TV}_\al(S) + \FRc \AIc \rB{\bCES{S^*}{S^*} - \aCTEa{S}} . 
\end{align*}
\end{corollary}

\begin{proof}
Simply substituting $k=2$ into Theorem \ref{thm:2} and noting that $\bTCMa{1}{S}=0$, we obtain the desired result.
\end{proof}
\begin{remark} \label{rem:5}
In recent literature, \cite{IL25} and \cite{Y25} studied $\Var[X_i \mid S>s_\al]$ and $\aTCOVS{X_i, X_j}$ respectively due to their relevance to the tail behaviour of $X_i$.  We provide the expressions for two relevant identities for the NMVM model, which are directly obtainable from Lemma \ref{lmm:4}, given by 
\begin{align*}
\bCES{X_i^2}{S} 
=&~ \AIb^2 \bCES{S^2}{S}
  + 2 \AIa \AIb \bCES{S}{S} \\
  &~~ + 2 \AIc \AIb \FRc \bCES{S^*}{S^*} + \AIa^2 \\
  &~~ + (2 \AIa \AIc + \SGi - \AIb^2 \SGS) \FRc 
  + \AIc^2 \FRd ,
\end{align*}
and 
\begin{align*}
\bCES{X_i X_j}{S} 
=&~ \AIb \AJb \bCES{S^2}{S}
  + \rB{\AIb \AJa + \AIa \AJb} \bCES{S}{S} \\ 
  &~~ + \rB{\AIb \AJc + \AIc \AJb} \FRc \bCES{S^*}{S^*} + \AIa \AJa \\ 
  &~~ + \rB{\AIc \AJa + \AIa \AJc + \SGij - \AIb \AJb \SGS} \FRc
  + \AIc \AJc \FRd .
\end{align*}
See Appendix \ref{app:B} for the derivation of these identities.
\end{remark}

\section{Numerical illustration} \label{sec:5}

This section applies the TCM-based capital allocation results obtained in previous sections to financial losses modelled by the multivariate generalised hyperbolic (GH) distribution. A capital allocation based on both the CTE and TCMs is used to decide an appropriate capital reserve allocation.

For this illustration, we selected the historical daily log losses of four stocks, namely Boeing (BA), American Express (AXP), ExxonMobil (XOM), and Chevron (CVX), denoted by $X_1,\dots,X_4$, from 1 January 2020 to 31 December 2024 (1257 trading days). The daily log loss of a stock at day $t$ is calculated as $L_t = - \ln \rB{P_t/P_{t-1}}$, where $P_t$ is the adjusted closing price of the stock at trading day $t$. Historical stock data are obtained from Yahoo Finance via the R package \texttt{quantmod}.

The summary statistics of the data are shown in Table \ref{tbl:1}. We observe that all stocks exhibit non-zero skewness and that the kurtosis is much greater than 3 (the kurtosis of the normal distribution). This indicates the existence of heavy tails in the data, which can be captured by the multivariate GH distribution. 

\begin{table}[H] 
\centering
\begin{tabular}{|l|c|c|c|c|c|c|c|}  
\hline
\textbf{Index} & \textbf{Mean} & \textbf{Median} & \textbf{Minimum} & \textbf{Maximum} & \textbf{St.Dev.} & \textbf{Skewness} & \textbf{Kurtosis} \\ \hline
\textbf{BA}     & 0.000501 & 0.000422 & -0.217677 & 0.272444 & 0.032270 & 0.421802 & 15.44124 \\ \hline
\textbf{AXP}    & -0.000737 & -0.000785 & -0.197886 & 0.160388 & 0.024025 & -0.599463 & 16.69053 \\ \hline
\textbf{XOM}    & -0.000511 & -0.000212 & -0.119442 & 0.130391 & 0.021658 & 0.161940 & 7.63877 \\ \hline
\textbf{CVX}    & -0.000308 & -0.000787 & -0.204904 & 0.250062 & 0.022591 & 1.072524 & 27.08356 \\ \hline
\end{tabular}
\caption{Descriptive statistics of the stocks’ daily log losses} \label{tbl:1}
\end{table}

To fit the multivariate GH distribution, we used the EM algorithm calibration in \cite{M15} implemented via the \texttt{fit.ghypmv} function in the R package \texttt{ghyp}. As our goal in this section is to illustrate the impact of incorporating the TCMs into the CTE-based capital allocation, we are not concerned with finding the best-fit model in the NMVM or GH families. For such empirical tasks, we refer to \cite{IL15} and \cite{IL19}. The fitted model is $\mathbf{X} \sim MGH_4(-1.689,\ 4.509 \times 10^{-5},\ 1.380,\ \bm{\mu},\ \bm{\Sigma},\ \bm{\gm})$, where
\begin{align*}
\bm{\mu}' 
=&~ \rB{2.393,\ -15.135,\ -0.474,\ -0.305} \times 10^{-4} ; \\
\bm{\gamma}' 
=&~ \rB{2.556,\ 7.584,\ -4.530,\ -0.0287} \times 10^{-4} ; \\
\bm{\Sigma} =&~ 
  \begin{pmatrix}
    9.462 & 3.790 & 2.710 & 2.538 \\
    3.790 & 5.278 & 2.533 & 2.417 \\
    2.710 & 2.533 & 5.495 & 4.338 \\
    2.538 & 2.417 & 4.338 & 4.413
  \end{pmatrix} \times 10^{-4}.
\end{align*}
The fitted density function of each marginal $X_i$ is shown below.
\begin{figure} [H]
    \centering
    \includegraphics[width=1\linewidth]{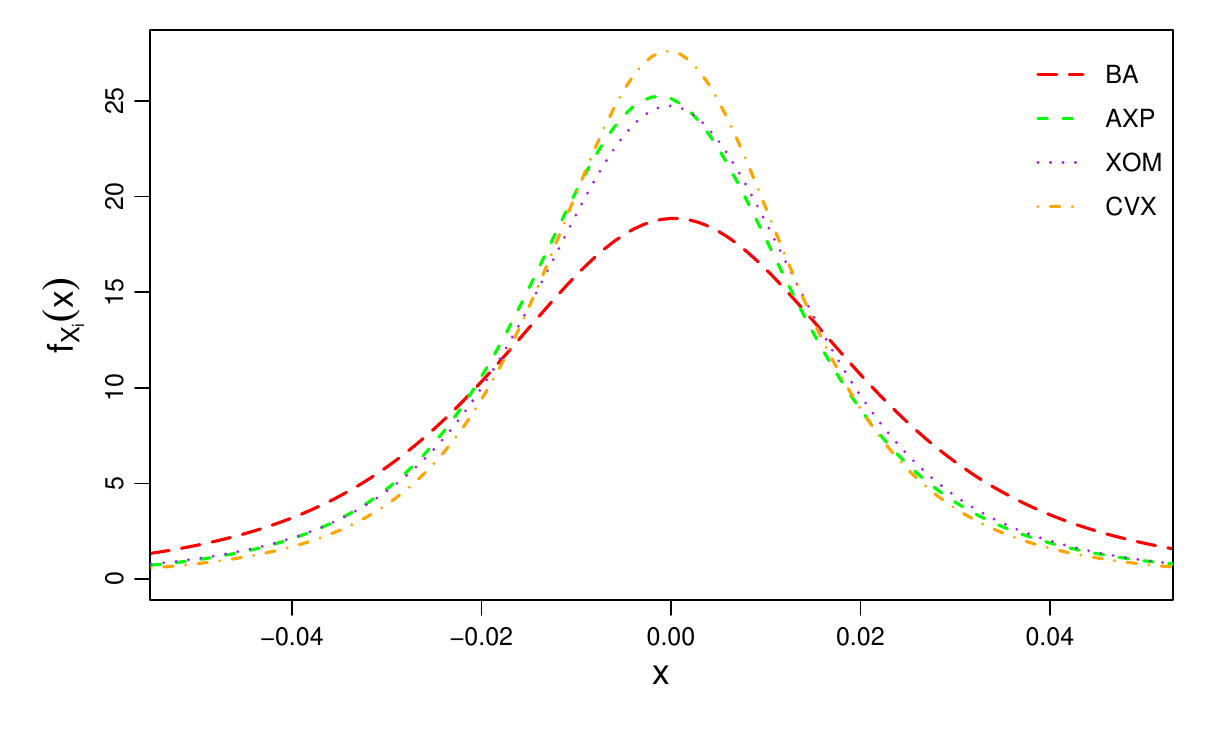}
    \caption{Marginal densities, $f_{X_i}(x)$, of $X_1, X_2, X_3, X_4$}
    \label{fig:1}
\end{figure}

From Figure \ref{fig:1}, it is seen that the log losses are slightly asymmetric in general. The stock losses are positively correlated as seen from the parameter $\bm{\Sigma}$, which is reasonable since companies such as XOM and CVX are from the same industry, and therefore the diversification effect is not as strong as expected for this portfolio. Among the individual stocks, BA has a positive mean log loss and a visibly heavier tail than the rest, indicating its riskiness as an investment choice.

Suppose that we have invested a total of \$100 equally distributed to $X_1$ to $X_4$. We write the total nominal loss of the portfolio as $S := w_1 X_1+w_2 X_2+w_3 X_3+w_4 X_4$ where $w_1, w_2, w_3, w_4$ represent the nominal amounts invested into each stock ($w_1=\dots=w_4=\$25$ in our scenario).  It is established that capital allocations based on the CTE, TV, and TCM$_{\al,3}$, respectively will yield the following allocation outcome:
\begin{enumerate}[(i)]
\item $K = \aCTEa{S}$ and $K_i = \bCES{w_i X_i}{S}$ for all $i=1,\dots,n$;
\item $K = \mathrm{TV}_\al(S)$ and $K_i = \aTCOVS{w_i X_i,S}$ for all $i=1,\dots,n$;
\item $K = \bTCMa{3}{S}$ and $K_i = \aTCOVS{w_i X_i, (S-\aCTEa{S})^2}$ for all $i=1,\dots,n$,
\end{enumerate}
where the capital allocated can be calculated by Lemma \ref{lmm:3}, Corollary \ref{cor:2}, and Theorem \ref{thm:2}.

Figure \ref{fig:2} below plots $\aCTEa{S}$, $\mathrm{TV}_\al(S)$, and $\bTCMa{3}{S}$ and their allocations to each stock. It also displays the relative proportions of the capital allocated (given by $K_i/K$), which can be interpreted as the risk contribution by each stock. Selected capital allocation values for some quantiles are also
presented in Table \ref{tbl:2}.
\begin{figure} [H]
    \centering
    \includegraphics[width=1\linewidth]{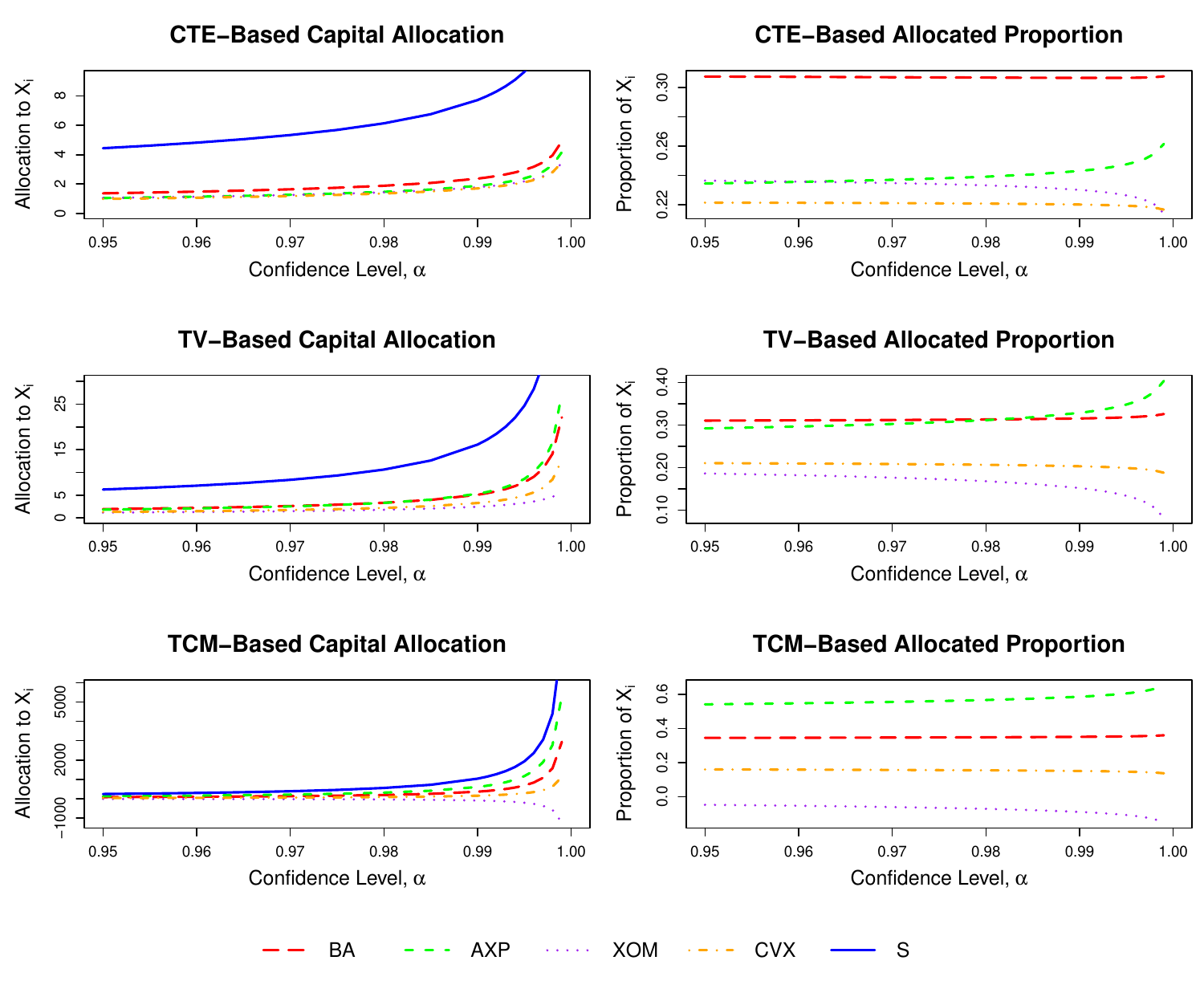}
    \caption{The capital allocated to $X_1, X_2, X_3, X_4$ based on the CTE, TV, and TCM$_{\al,3}$, and their relative proportions}
    \label{fig:2}
\end{figure} 
\begin{table}[ht]
\centering
\begin{tabular}{|c|c|c|c|c|c|}
\hline
$\bm \al$ & \textbf{Method} & \textbf{BA} & \textbf{AXP} & \textbf{XOM} & \textbf{CVX} \\
\hline
~~~~~ & CTE           & 1.367 & 1.042 & 1.051 & 0.984 \\
0.950 & TV            & 1.941 & 1.826 & 1.165 & 1.317 \\
~~~~~ & TCM$_{\al,3}$ & 84.616 & 132.798 & -11.467 & 39.308 \\
\hline
~~~~~ & CTE           & 1.482 & 1.136 & 1.137 & 1.067 \\
0.960 & TV            & 2.208 & 2.105 & 1.293 & 1.489 \\
~~~~~ & TCM$_{\al,3}$ & 103.534 & 163.942 & -15.735 & 47.600 \\
\hline
~~~~~ & CTE           & 1.640 & 1.266 & 1.254 & 1.181 \\
0.970 & TV            & 2.614 & 2.536 & 1.480 & 1.748 \\
~~~~~ & TCM$_{\al,3}$ & 134.353 & 215.154 & -23.241 & 60.949 \\
\hline
~~~~~ & CTE           & 1.884 & 1.468 & 1.432 & 1.356 \\
0.980 & TV            & 3.331 & 3.314 & 1.790 & 2.199 \\
~~~~~ & TCM$_{\al,3}$ & 194.097 & 315.683 & -39.261 & 86.399 \\
\hline
~~~~~ & CTE           & 2.369 & 1.878 & 1.778 & 1.701 \\
0.990 & TV            & 5.091 & 5.303 & 2.456 & 3.280 \\
~~~~~ & TCM$_{\al,3}$ & 364.372 & 608.061 & -91.787 & 156.936 \\
\hline
~~~~~ & CTE           & 4.918 & 4.180 & 3.422 & 3.463 \\
0.999 & TV            & 22.113 & 27.396 & 5.563 & 12.761 \\
~~~~~ & TCM$_{\al,3}$ & 2940.939 & 5323.095 & -1227.183 & 1125.261 \\
\hline
\end{tabular}
\caption{Capital allocated to $X_1, X_2, X_3, X_4$ based on the CTE, TV, and TCM$_{\al,3}$} \label{tbl:2}
\end{table}
The allocated proportions to BA and CVX remain stable over all $\al\in(0.95,1)$ and for the three allocation methods based on the CTE, TV, and TCM$_{\al,3}$, but they are very different for AXP and XOM. When $\al<0.98$, the allocated proportion to AXP for the TV is noticeably higher than for the CTE (increasing from approximately 24\% to 30\% of the total). This trend persists when switching from the TV to the TCM$_{\al,3}$. This observation flips for XOM. Interestingly, the risk contribution to the TCM$_{\al,3}$ for XOM is negative, indicating some diversification benefit to the portfolio. When $\al>0.98$, the TV and TCM$_{\al,3}$ amplify the tail behaviour of AXP and XOM (relative to the CTE) to different extents. This is sensible as the TV and TCM measure different dependencies between $X_i$ and $S$, namely the expectation and dispersion in the tail region, respectively. This demonstrates the necessity of including the TV and TCM$_{\al,3}$ for a more thorough understanding of the stocks' tail behaviour. 

The observations so far suggest that neither the CTE-based nor TCM-based capital allocation should be used in isolation. Therefore, we suggest a linear combination of the CTE, TV, and TCM$_{\al,3}$, as previously mentioned. By taking $m_1=1$, $m_2=p$ and $m_3=q$ in \eqref{eq:2}, the total capital reserve is given by
\begin{align} \label{eq:14}
K = \aCTEa{S} + p \cdot \mathrm{TV_\al}(S) + q \cdot \bTCMa{3}{S},
\end{align} 
where $p,~q \ge 0$ represent the importance of the TV and TCM$_{\al,3}$ relative to the CTE, and $\al \in (0,1)$ is the confidence level. The corresponding capital allocated to stock $i$ for $i=1,\dots,4$ are given by
\begin{align*}
K_i = \bCES{w_i X_i}{S} + p \cdot \aTCOVS{w_i X_i,S} + q \cdot \aTCOVS{w_i X_i, (S-\aCTEa{S})^2}.
\end{align*}
To ensure each term in \eqref{eq:14} has a similar magnitude based on their values in Table \ref{tbl:2}, a reasonable choice for $p$ and $q$ is to select $p \in [0, 3]$ and $q \in [0, 0.005]$. Figure \ref{fig:3} shows how the capital allocation varies when priority shifts from the CTE to the TV and TCM$_{\al,3}$, as shown by different selections of $p$ and $q$, and Figure \ref{fig:4} shows the respective proportions of allocated capitals.
\begin{figure} [H]
    \centering
    \includegraphics[width=1\linewidth]{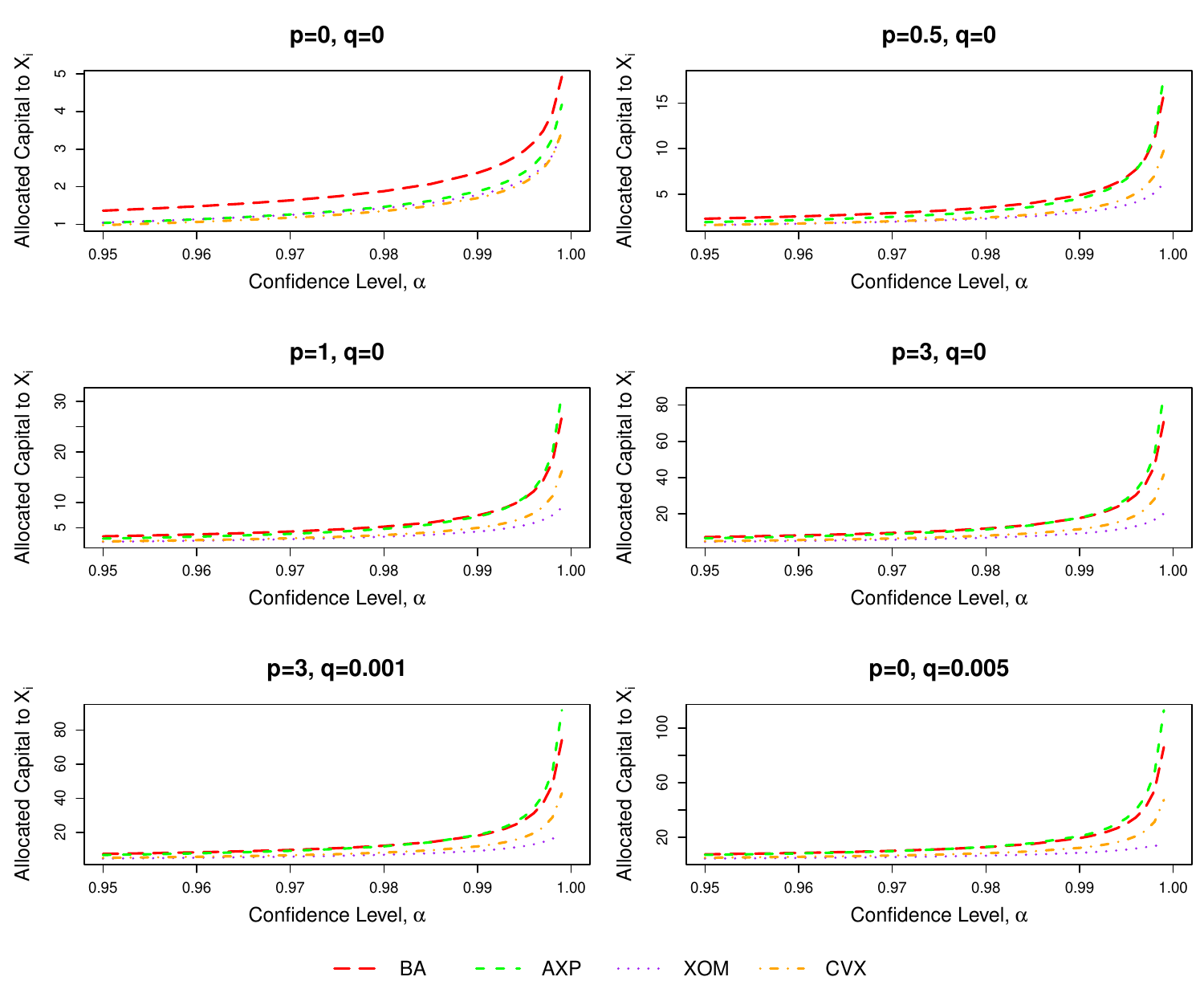}
    \caption{Capital allocated under different CTE-based, TV-based, and TCM$_{\al,3}$-based capital allocation combinations}
    \label{fig:3} 
\end{figure}
\begin{figure} [H]
    \centering  
    \includegraphics[width=1\linewidth]{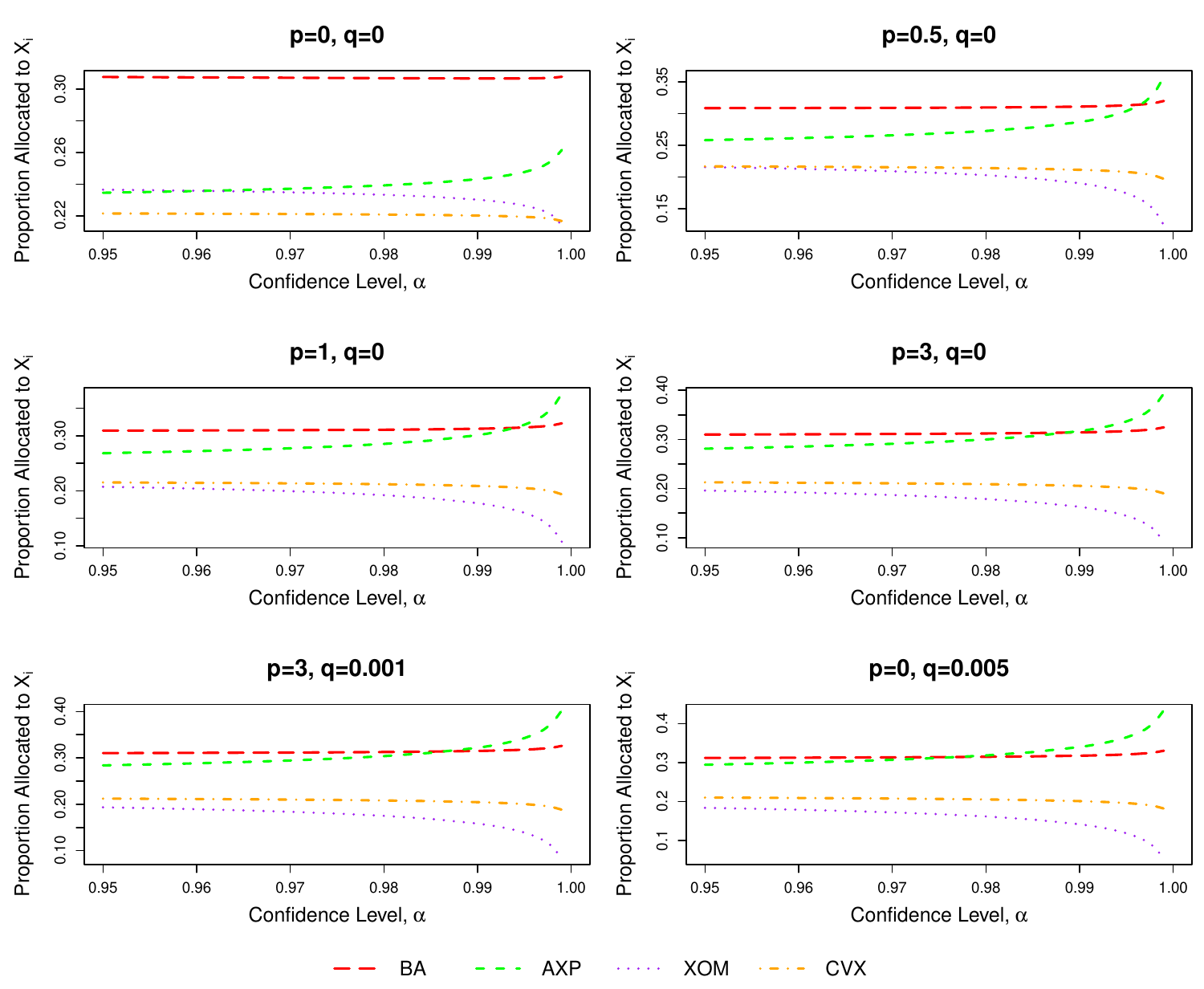}
    \caption{Proportions of capital allocated under different CTE-based, TV-based, and TCM$_{\al,3}$-based capital allocation combinations}
    \label{fig:4}
\end{figure}
The overall observations are not too surprising, as the individual patterns are already displayed in Figure \ref{fig:2}. The more priority placed on the TV or TCM$_{\al,3}$, the more capital allocated for AXP, the lesser for XOM, and roughly the same for BA and CVX.

\section{Conclusion} \label{sec:6}
In this paper, we introduce a new capital allocation method based on the tail central moments (TCM), which includes the tail variance-based capital allocation of \cite{V04} and \cite{FL06}. Together with the conditional tail expectation (CTE)-based capital allocation, the TCM-based capital allocation provides a more thorough risk assessment approach. This method is applied to the class of normal mean–variance mixture (NMVM) distributions, which has widespread finance and insurance applications. In particular, we derive analytical recursive expressions for the TCM and its capital allocation for the multivariate NMVM class. A numerical illustration shows that the TCM is an insightful risk metric that reveals important tail behaviours which are otherwise not detectable by the CTE alone. These results provide a readily applicable framework to assess each component's risk contribution to the portfolio's total risk and to quantify interconnected risks, enabling financial and insurance agents to reliably estimate their tail losses.

\bibliographystyle{apalike} \bibliography{bibliography}

@article{A99,
  title={Coherent measures of risk},
  author={Artzner, Philippe and Delbaen, Freddy and Eber, Jean-Marc and Heath, David},
  journal={Mathematical Finance},
  volume={9},
  number={3},
  pages={203--228},
  year={1999},
  publisher={Wiley Online Library}
}

@article{AP18,
  title={A polynomial goal programming model for portfolio optimization based on entropy and higher moments},
  author={Aksarayl{\i}, Mehmet and Pala, Osman},
  journal={Expert Systems with Applications},
  volume={94},
  pages={185--192},
  year={2018},
  publisher={Elsevier}
}

@article{D01,
  title={Coherent allocation of risk capital},
  author={Denault, Michel},
  journal={Journal of Risk},
  volume={4},
  pages={1--34},
  year={2001}
}

@article{D12,
  title={Optimal capital allocation principles},
  author={Dhaene, Jan and Tsanakas, Andreas and Valdez, Emiliano and Vanduffel, Steven},
  journal={Journal of Risk and Insurance},
  volume={79},
  number={1},
  pages={1--28},
  year={2012},
  publisher={Wiley Online Library}
}

@article{D20,
  title={Size-biased risk measures of compound sums},
  author={Denuit, Michel},
  journal={North American Actuarial Journal},
  volume={24},
  number={4},
  pages={512--532},
  year={2020},
  publisher={Taylor \& Francis}
}

@article{EK95,
  title={Hyperbolic distributions in finance},
  author={Eberlein, Ernst and Keller, Ulrich},
  journal={Bernoulli},
  pages={281--299},
  year={1995},
  publisher={JSTOR}
}

@article{EK21,
  title={Tail conditional moment for generalized skew-elliptical distributions},
  author={Eini, Esmat Jamshidi and Khaloozadeh, Hamid},
  journal={Journal of Applied Statistics},
  volume={48},
  number={13-15},
  pages={2285--2305},
  year={2021},
  publisher={Taylor \& Francis}
}

@article{FL06,
  title={Tail variance premium with applications for elliptical portfolio of risks},
  author={Furman, Edward and Landsman, Zinoviy},
  journal={ASTIN Bulletin},
  volume={36},
  number={2},
  pages={433--462},
  year={2006},
  publisher={Cambridge University Press}
}

@article{FL10,
  title={Multivariate {T}weedie distributions and some related capital-at-risk analyses},
  author={Furman, Edward and Landsman, Zinoviy},
  journal={Insurance: Mathematics and Economics},
  volume={46},
  number={2},
  pages={351--361},
  year={2010},
  publisher={Elsevier}
}

@article{G22,
  title={The gradient allocation principle based on the higher moment risk measure},
  author={G{\'o}mez, Fabio and Tang, Qihe and Tong, Zhiwei},
  journal={Journal of Banking \& Finance},
  volume={143},
  pages={106544},
  year={2022},
  publisher={Elsevier}
}

@article{G21,
  title={Capital allocation techniques: review and comparison},
  author={Guo, Qiheng and Bauer, Daniel and Zanjani, George},
  journal={Variance},
  volume={14},
  number={2},
  year={2021}
}

@article{H10,
  title={Portfolio selection with higher moments},
  author={Harvey, Campbell R and Liechty, John C and Liechty, Merrill W and M{\"u}ller, Peter},
  journal={Quantitative Finance},
  volume={10},
  number={5},
  pages={469--485},
  year={2010},
  publisher={Taylor \& Francis}
}

@article{H19,
  title={Extreme conditional tail moment estimation under serial dependence},
  author={Hoga, Yannick},
  journal={Journal of Financial Econometrics},
  volume={17},
  number={4},
  pages={587--615},
  year={2019},
  publisher={Oxford University Press}
}

@article{IL15,
  title={Estimating the tails of loss severity via conditional risk measures for the family of symmetric generalised hyperbolic distributions},
  author={Ignatieva, Katja and Landsman, Zinoviy},
  journal={Insurance: Mathematics and Economics},
  volume={65},
  pages={172--186},
  year={2015},
  publisher={Elsevier}
}

@article{IL19,
  title={Conditional tail risk measures for the skewed generalised hyperbolic family},
  author={Ignatieva, Katja and Landsman, Zinoviy},
  journal={Insurance: Mathematics and Economics},
  volume={86},
  pages={98--114},
  year={2019},
  publisher={Elsevier}
}

@article{IL21,
  title={A class of generalised hyper-elliptical distributions and their applications in computing conditional tail risk measures},
  author={Ignatieva, Katja and Landsman, Zinoviy},
  journal={Insurance: Mathematics and Economics},
  volume={101},
  pages={437--465},
  year={2021},
  publisher={Elsevier}
}

@article{IL25,
  title={Tail variance for generalised hyper-elliptical models},
  author={Ignatieva, Katja and Landsman, Zinoviy},
  journal={ASTIN Bulletin},
  volume={55},
  number={1},
  pages={144--167},
  year={2025},
  publisher={Cambridge University Press}
}

@article{K05,
  title={An axiomatic approach to capital allocation},
  author={Kalkbrener, Michael},
  journal={Mathematical Finance},
  volume={15},
  number={3},
  pages={425--437},
  year={2005},
  publisher={Wiley Online Library}
}

@article{K07,
  title={Higher moment coherent risk measures},
  author={Krokhmal, Pavlo A},
  journal={Quantitative Finance},
  volume={7},
  number={4},
  pages={373--387},
  year={2007},
  publisher={Taylor \& Francis}
}

@article{K10,
  title={Conditional tail moments of the exponential family and its related distributions},
  author={Kim, Joseph HT},
  journal={North American Actuarial Journal},
  volume={14},
  number={2},
  pages={198--216},
  year={2010},
  publisher={Taylor \& Francis}
}

@article{KK19,
  title={Tail risk measures and risk allocation for the class of multivariate normal mean-variance mixture distributions},
  author={Kim, Joseph HT and Kim, So-Yeun},
  journal={Insurance: Mathematics and Economics},
  volume={86},
  pages={145--157},
  year={2019},
  publisher={Elsevier}
}

@article{KW14,
  title={A comparison of generalized hyperbolic distribution models for equity returns},
  author={Socgnia, Virginie Konlack and Wilcox, Diane},
  journal={Journal of Applied Mathematics},
  volume={2014},
  number={1},
  pages={263465},
  year={2014},
  publisher={Wiley Online Library}
}

@article{L13,
  title={Tail variance premiums for log-elliptical distributions},
  author={Landsman, Zinoviy and Pat, Nika and Dhaene, Jan},
  journal={Insurance: Mathematics and Economics},
  volume={52},
  number={3},
  pages={441--447},
  year={2013},
  publisher={Elsevier}
}

@article{L23,
  title={Asymptotic results on tail moment and tail central moment for dependent risks},
  author={Li, Jinzhu},
  journal={Advances in Applied Probability},
  volume={55},
  number={4},
  pages={1116--1143},
  year={2023},
  publisher={Cambridge University Press}
}

@article{LG04,
  title={An optimization approach to the dynamic allocation of economic capital},
  author={Laeven, R. J. and Goovaerts, M. J.},
  journal={Insurance: Mathematics and Economics},
  volume={35},
  number={3},
  pages={299--319},
  year={2004}
}

@article{LV03,
  title={Tail conditional expectations for elliptical distributions},
  author={Landsman, Zinoviy and Valdez, Emiliano},
  journal={North American Actuarial Journal},
  volume={7},
  number={4},
  pages={55--71},
  year={2003},
  publisher={Taylor \& Francis}
}

@article{LV16,
  title={The tail Stein's identity with applications to risk measures},
  author={Landsman, Zinoviy and Valdez, Emiliano},
  journal={North American Actuarial Journal},
  volume={20},
  number={4},
  pages={313--326},
  year={2016},
  publisher={Taylor \& Francis}
}

@book{M15,
  title = {Quantitative Risk Management: Concepts, Techniques, and Tools - Revised Edition},
  author = {McNeil, Alexander J and Frey, R{\"u}diger and Embrechts, Paul},
  year = {2015},
  publisher = {Princeton University Press}
}

@article{N09,
  title={Modeling heavy-tailed stock index returns using the generalized hyperbolic distribution},
  author={Necula, Ciprian},
  journal={Romanian Journal of Economic Forecasting},
  volume={10},
  number={2},
  pages={118--131},
  year={2009}
}

@incollection{O00,
  author={Overbeck, Ludger},
  title={Allocation of Economic Capital in Loan Portfolios},
  booktitle={Measuring Risk in Complex Systems},
  editor={Franke, J{\"u}rgen and H{\"a}rdle, Wolfgang and Stahl, Gero},
  pages={1--18},
  year={2000},
  publisher={Springer},
  address={New York}
}

@book{P02,
  title={Measurement of Risk, Solvency Requirements and Allocation of Capital Within Financial Conglomerates},
  author={Panjer, Harry H},
  year={2002},
  publisher={University of Waterloo, Institute of Insurance and Pension Research Waterloo}
}

@article{R22,
  title={Tail Moments of Compound Distributions},
  author={Ren, Jiandong},
  journal={North American Actuarial Journal},
  volume={26},
  number={3},
  pages={336--350},
  year={2022},
  publisher={Taylor \& Francis}
}

@article{RU00,
  title={Optimization of conditional Value-at-Risk},
  author={Rockafellar, R Tyrrell and Uryasev, Stanislav},
  journal={Journal of Risk},
  volume={2},
  pages={21--42},
  year={2000}
}

@article{S01,
  title={Markowitz revisited: mean-variance models in financial portfolio analysis},
  author={Steinbach, Marc C},
  journal={SIAM Review},
  volume={43},
  number={1},
  pages={31--85},
  year={2001},
  publisher={SIAM}
}

@article{S70,
  title={The fundamental approximation theorem of portfolio analysis in terms of means, variances and higher moments},
  author={Samuelson, Paul A},
  journal={The Review of Economic Studies},
  volume={37},
  number={4},
  pages={537--542},
  year={1970},
  publisher={Wiley-Blackwell}
}

@article{T01,
  title={Conditional expectation as quantile derivative},
  author={Tasche, Dirk},
  journal={arXiv preprint math/0104190},
  year={2001}
}

@article{T04,
  title={Allocating portfolio economic capital to sub-portfolios},
  author={Tasche, Dirk},
  journal={Economic Capital: a Practitioner Guide},
  pages={275--302},
  year={2004},
  publisher={Risk books London}
}

@incollection{T08,
  author={Dirk Tasche},
  title={Capital allocation to business units and sub-portfolios: the {E}uler principle},
  booktitle={Pillar II in the New Basel Accord: the Challenge of Economic Capital},
  editor={Andrea Resti},
  publisher={Risk Books},
  year={2008},
  pages={423--453}
}

@book{T12,
  title={The Multivariate Normal Distribution},
  author={Tong, Yung Liang},
  year={2012},
  publisher={Springer}
}

@article{V04,
  title={On tail conditional variance and tail covariances},
  author={Valdez, Emiliano},
  journal={UNSW Actuarial Studies, Sydney},
  year={2004}
}

@article{V05,
  title={Tail conditional variance for elliptically contoured distributions},
  author={Valdez, Emiliano},
  journal={Belgian Actuarial Bulletin},
  volume={5},
  number={1},
  pages={26--36},
  year={2005}
}

@article{V06,
  title={Multivariate skew-normal distributions with applications in insurance},
  author={Vernic, Raluca},
  journal={Insurance: Mathematics and Economics},
  volume={38},
  number={2},
  pages={413--426},
  year={2006},
  publisher={Elsevier}
}

@article{W14,
  title={Capital allocation based on the tail covariance premium adjusted},
  author={Wang, Min},
  journal={Insurance: Mathematics and Economics},
  volume={57},
  pages={125--131},
  year={2014},
  publisher={Elsevier}
}

@article{Y25,
  title={Tail moments and tail joint moments for multivariate generalized hyperbolic distribution},
  author={Yang, Yang and Wang, Guojing and Yao, Jing},
  journal={Journal of Computational and Applied Mathematics},
  volume={457},
  pages={116307},
  year={2025},
  publisher={Elsevier}
}

@book{J82,
  title={Statistical Properties of the Generalized Inverse Gaussian Distribution},
  author={J{\o}rgensen, Bent},
  volume={9},
  year={1982},
  publisher={Lecture Notes in Statistics}
}

\newpage
\appendix

% Reset all counters

\setcounter{table}{0}
\setcounter{figure}{0}
\setcounter{equation}{0}
\renewcommand{\thetable}{A.\arabic{table}}
\renewcommand{\thefigure}{A.\arabic{figure}}
\renewcommand{\theequation}{A.\arabic{equation}}

\setcounter{theorem}{0}
\setcounter{proposition}{0}
\renewcommand{\thetheorem}{A.\arabic{theorem}}
\renewcommand{\theproposition}{A.\arabic{proposition}}
\setcounter{lemma}{0}
\renewcommand{\thelemma}{A.\arabic{lemma}}
\setcounter{example}{0}
\renewcommand{\theexample}{A.\arabic{example}}
\setcounter{corollary}{0}
\renewcommand{\thecorollary}{A.\arabic{corollary}}
\setcounter{remark}{0}
\renewcommand{\theremark}{A.\arabic{remark}}
\setcounter{definition}{0}
\renewcommand{\thedefinition}{A.\arabic{definition}}

\begin{center}
\LARGE	Appendices
\end{center}

\section{The TCM-based Euler allocation principle} \label{app:A}

This section derives the TCM-based capital allocation using the Euler allocation principle in Remark \ref{rem:3}. For $\bm{w}=(w_1,\dots,w_n) \in \R^n$, define  $L(\bm{w})=w_1 X_1 + \dots + w_n X_n$, and the aggregate loss $S=L(1, \dots, 1)$. Denote by $l_\al(\bm{w})$ the $\al$-quantile of $L(\bm{w})$ for $\al \in (0,1)$. A risk measure is a functional that maps random variables to the real line.  A risk measure $\rho$ is positive homogeneous if  for all $t>0$ and any random variable $X$,  $\rho(tX))=t\rho(X)$. Assuming that  $\rho$ is positive homogeneous and  $\rho(L(\bm{w}))$ is  continuously differentiable in $\bm{w} \in \R^n$, the Euler allocation principle with risk measure $\rho$ is defined as 
\begin{align*}
K=  \rho(L(1,\dots,1)) \mbox{~~and~~} K_i = 
w_i \left. \frac{\partial \rho(L(\bm{w}))}{\partial w_i} \right|_{\bm{w=1}},
\end{align*}
where $K$ is the capital reserve for $S$ and $K_i$ is the capital allocated to $X_i$. The Euler allocation principle automatically satisfies the full allocation property since
\begin{align*}
\rho(L(\bm{w})) =  \iSum w_i \frac{\partial \rho(L(\bm{w}))}{\partial w_i} \mbox{~~holds for all $\bm{w}\in\R^n$}.
\end{align*} 

%
%In our scenario, the overall capital reserve is $K=\rho(S)$, and the capital allocated to component $X_i$ under the Euler allocation principle is given by  
%%
%\begin{align*}
%K_i = 
%w_i \left. \frac{\partial \rho(L(\bm{w}))}{\partial w_i} \right|_{\bm{w=1}} \mbox{~~with~~} w_i=1 .
%\end{align*}
%%
Remark \ref{rem:3} states that the Euler allocation method is not applicable to the total capital reserve $\rho(S)=\bTCMa{k}{S}$ as in Definition \ref{def:3}. This is because the TCM is not positive homogeneous, and some modifications are required.
% Here, we revisit this Remark \ref{rem:3} in more detail.
%
\begin{proposition}
Fix $\al \in (0,1)$ and $k \in \N$. Assume that the random vector $(X_1, \dots, X_n) \in \R^n$ satisfies Assumption 2.3 of \cite{T01}. The Euler allocation principle with $\bTCMa{k}{\cdot}^{1/k}$ is  given by
\begin{align*}
K = \bTCMa{k}{S}^{\frac{1}{k}} \mbox{~~and~~}
K_i = \frac{ \aTCOVS{X_i, (S-\aCTEa{S})^{k-1} } }{
    \bTCMa{k}{S}^{1-\frac{1}{k}} } .
\end{align*}
\end{proposition}

\begin{proof} 
It is easy to show that $\bTCMa{k}{\cdot}^{1/k}$ is partially differentiable (refer to, e.g., \cite{T01}) and positive homogeneous. We first require Corollary 4.2 of \cite{T01}, which states that 
\begin{align*}
\frac{\partial}{\partial w_i} \E \sB{L(\bm{w})^k \mid L(\bm{w}) \geq l_\al(\bm{w})} = k \E \sB{X_i L(\bm{w})^{k-1} \mid L(\bm{w}) \geq l_\al(\bm{w})} .
\end{align*}
For $\bm{w} \in \R^n$, denote by $\rho^*(\bm{w})=\bTCMa{k}{L(\bm{w})}$ and $\rho(\bm{w})=\bTCMa{k}{L(\bm{w})}^{\frac{1}{k}}$. Using the above result gives
\begin{align} \label{eq:A1}
\left. \frac{\partial \rho^*(\bm{w})}{\partial w_i} \right|_{\bm{w=1}}
=&~ \left. \frac{\partial}{\partial w_i} 
    \E \sB{(L(\bm{w})-\aCTEa{L(\bm{w})})^k \mid L(\bm{w}) > l_\al(\bm{w})}
  \right|_{\bm{w=1}} \nonumber \\
=&~ \left. \frac{\partial}{\partial w_i} \rB{
    \sum_{j=0}^{k} \binom{k}{j} (-1)^j \aCTEa{L(\bm{w})}^j 
    \E \sB{L(\bm{w})^{k-j} \mid L(\bm{w}) > l_\al(\bm{w})} 
  } \right|_{\bm{w=1}} \nonumber \\
=&~ \sum_{j=0}^{k} \binom{k}{j} (-1)^j 
  \left. \frac{\partial}{\partial w_i} \rB{ \aCTEa{L(\bm{w})}^j
     \E \sB{L(\bm{w})^{k-j} \mid L(\bm{w}) > l_\al(\bm{w})} 
  } \right|_{\bm{w=1}} ,
\end{align}
where
\begin{align*}
&~ \left. \frac{\partial}{\partial w_i} \rB{ \aCTEa{L(\bm{w})}^j
    \E \sB{L(\bm{w})^{k-j} \mid L(\bm{w}) > l_\al(\bm{w})} 
  } \right|_{\bm{w=1}} \\
=&~ \aCTEa{L(\bm{w})}^j 
  \left. \frac{\partial}{\partial w_i} 
    \E \sB{L(\bm{w})^{k-j} \mid L(\bm{w}) > l_\al(\bm{w})}
  \right|_{\bm{w=1}}
\\ &~
  + \E \sB{L(\bm{w})^{k-j} \mid L(\bm{w}) > l_\al(\bm{w})} \left. \frac{\partial}{\partial w_i} 
    \rB{ \E \sB{L(\bm{w}) \mid L(\bm{w}) > l_\al(\bm{w})} }^j 
  \right|_{\bm{w=1}} \\
=&~ \aCTEa{S}^j \cdot (k-j) \bCES{X_i S^{k-j-1}}{S}
  + \bCES{S^{k-j}}{S} \cdot j ~ \aCTEa{S}^{j-1} \bCES{X_i}{S} .
\end{align*}
Hence, \eqref{eq:A1} becomes
\begin{align*}
&~ \sum_{j=0}^{k} \binom{k}{j} (k-j) (-1)^j \aCTEa{S}^j 
  \bCES{X_i S^{k-j-1}}{S}
\\ &~
  + \sum_{j=0}^{k} \binom{k}{j} (j) (-1)^j \aCTEa{S}^{j-1} \bCES{X_i}{S} \bCES{S^{k-j}}{S} \\
=&~ \sum_{j=0}^{k-1} \binom{k}{j} (k-j) (-1)^j \aCTEa{S}^j 
  \bCES{X_i S^{k-j-1}}{S} + 0
\\ &~
  + 0 + \sum_{j=1}^{k} \binom{k}{j} (j) (-1)^j \aCTEa{S}^{j-1} \bCES{X_i}{S} \bCES{S^{k-j}}{S} \\
=&~ k \sum_{j=0}^{k-1} \binom{k-1}{j} (-1)^j \aCTEa{S}^j 
  \bCES{X_i S^{k-j-1}}{S}
\\ &~
  + \sum_{j=0}^{k-1} \binom{k}{j+1} (j+1) (-1)^{j+1} \aCTEa{S}^j \bCES{X_i}{S} \bCES{S^{k-j-1}}{S} \\
=&~ k \sum_{j=0}^{k-1} \binom{k-1}{j} (-1)^j \aCTEa{S}^j 
  \bCES{X_i S^{k-j-1}}{S}
\\ &~
  - k \sum_{j=0}^{k-1} \binom{k-1}{j} (-1)^j \aCTEa{S}^j \bCES{X_i}{S} \bCES{S^{k-j-1}}{S} \\
=&~ k \bCES{X_i \sum_{j=0}^{k-1} \binom{k-1}{j} (-\aCTEa{S})^j S^{k-1-j}}{S}
\\ &~
  - k \bCES{X_i}{S} \bCES{\sum_{j=0}^{k-1} \binom{k-1}{j} (-\aCTEa{S})^j S^{k-1-j}}{S} \\
=&~ k \cTCMb{X_i}{S}{k-1} - k \bCES{X_i}{S} \aTCMS{k-1} \\
=&~ k \aTCOVS{X_i, (S-\aCTEa{S})^{k-1}} .
\end{align*}
Finally, the capital allocated to each component is given by
\begin{align*}
K_i 
= \left. \frac{\partial \rho(\bm{w})}{\partial w_i} \right|_{\bm{w=1}}
= \frac{1}{k ~ \bTCMa{k}{S}^{1-\frac{1}{k}}} \left. \frac{\partial \rho^*(\bm{w})}{\partial w_i} \right|_{\bm{w=1}} 
= \frac{ \aTCOVS{X_i, (S-\aCTEa{S})^{k-1} } }{
    \bTCMa{k}{S}^{1-\frac{1}{k}} } .
\end{align*}
The proof is complete.
\end{proof}

\section{Proof for Remark \ref{rem:5}} \label{app:B}

We revisit the identities given in Remark \ref{rem:5}, which is given below, in more detail.
\begin{lemma}
Consider the same random variables $X_1, \dots, X_n$, and $S$ in Lemma \ref{lmm:3}, as well as all related parameters and coefficients. We have the following identities:
\begin{align} \label{eq:A2}
\bCES{X_i^2}{S} 
=&~ \AIb^2 \bCES{S^2}{S}
  + 2 \AIa \AIb \bCES{S}{S} \nonumber \\
  &~~ + 2 \AIc \AIb \FRc \bCES{S^*}{S^*} + \AIa^2 \nonumber \\
  &~~ + (2 \AIa \AIc + \SGi - \AIb^2 \SGS) \FRc 
  + \AIc^2 \FRd ,
\end{align}
and 
\begin{align} \label{eq:A3}
\bCES{X_i X_j}{S} 
=&~ \AIb \AJb \bCES{S^2}{S}
  + \rB{\AIb \AJa + \AIa \AJb} \bCES{S}{S}
\nonumber \\ &~~
  + \rB{\AIb \AJc + \AIc \AJb} \FRc \bCES{S^*}{S^*} + \AIa \AJa 
\nonumber \\ &~~
  + \rB{\AIc \AJa + \AIa \AJc + \SGij - \AIb \AJb \SGS} \FRc
  + \AIc \AJc \FRd .
\end{align}
\end{lemma}
\begin{proof}
Before proving the lemma, we first provide a useful intermediate result below. Fix $k \in \N$, $l \in \N_0$, $\al \in (0,1)$, and let random variable $\Th^{*(l)}$ has density $\pi^{*(l)}(\th)$, with $\th>0$. We have
\begin{align} \label{eq:A4}
\aINTSa{ \INTth{ s^k \th^l \PIt \FSt } } 
=&~ \Cl \INTth{ \pi^{*(l)}(\th) \aINTSa{ s^k \FSt } } \nonumber \\
=&~ \Cl \INTth{ \pi^{*(l)}(\th) \aINTSa{ s^k f_{\SL | \th}(s) } } \nonumber \\
=&~ (1 - \ALl) \Cl \INTth{ \pi^{*(l)}(\th) \E\sB{(\SL)^k \mid \SL>s_\al, \Th^{*(l)}=\th} } \nonumber \\
=&~ (1 - \ALl) \Cl \bCES{(\SL)^k}{\SL} ,
\end{align}
where the second equality is due to $f_{\SL | \Th^{*(l)}}(s \mid \th)=f_{S | \Th^{*(l)}}(s \mid \th)$, based on the definition of $\SL$. 

Using (36) of \cite{KK19} (and directly replacing $X_i$ with $X_i X_j$), we obtain
\begin{align} \label{eq:A5}
\bCES{X_i X_j}{S} 
=&~ \FRa \aINTSa{ \aCEEs{X_i X_j} \FS } \nonumber \\
=&~ \FRa \aINTSa{ \INTth{ \aCEEst{X_i X_j} \FSt \PIt } } .
\end{align}
%
% Let $U=X_i X_i X_j$. We have
% 
% \begin{align*}
% \bCES{U}{S} 
% =&~ \FRa \aINTSa{ \aCEEs{U} \FS } \\
% =&~ \FRa \aINTSa{ \int^\oo_\oo u \frac{f_{U, S}(u, s)}{\FS} du \FS } \\
% =&~ \FRa \aINTSa{ \int^\oo_\oo \INTth{ u f_{U, S | \th}(u, s \mid \th) \PIt } du } \\
% =&~ \FRa \aINTSa{ \int^\oo_\oo \INTth{
%     u f_{U | s, \th}(u \mid s, \th) f_{\St}(s \mid \th) \PIt } du } \\
% =&~ \FRa \aINTSa{ \INTth{ \rB{ \int^\oo_\oo u f_{U | s, \th}(u \mid s, \th) du 
%   } f_{\St}(s \mid \th) \PIt } } \\
% =&~ \FRa \aINTSa{ \INTth{ \aCEEst{U} \PIt \FSt } } \\
% =&~ \FRa \aINTSa{ \INTth{ \aCEEst{X_i X_j} \PIt \FSt } } .
% \end{align*}
%
On the other hand, Lemma \ref{lmm:4} implies that
\begin{align*}
\aCEEst{X_i X_j} 
=&~ \aCEEst{X_i} \aCEEst{X_j} + \aCOVst{X_i, X_j} \\
=&~ \rB{\AIa + \AIc \th  + \AIb s} \rB{\AJa + \AJc \th + \AJb s}
  + \th (\SGij - \AIb \AJb \SGS) \\
=&~ \AIb \AJb s^2 
  + \rB{\AIb \AJa + \AIa \AJb} s
  + \rB{\AIb \AJc + \AIc \AJb} \th s
\\ &~~
  + \AIc \AJc \th^2
  + \rB{\AIc \AJa + \AIa \AJc + \SGij - \AIb \AJb \SGS} \th
  + \AIa \AJa .
\end{align*}
Substituting the above result into \eqref{eq:A5} and applying \eqref{eq:A4} gives
\begin{align*}
&~ \FRa \aINTSa{ \INTth{ \aCEEst{X_i X_j} \PIt } \FS } \\
=&~ \FRa \aINTSa{ \INTth{ \big( 
    \AIb \AJb s^2 
    + \rB{\AIb \AJa + \AIa \AJb} s
    + \rB{\AIb \AJc + \AIc \AJb} \th s
\\ &~~~~~~~~~~~~~~~~~~~~~~
    + \AIc \AJc \th^2
    + \rB{\AIc \AJa + \AIa \AJc + \SGij - \AIb \AJb \SGS} \th
    + \AIa \AJa \big) \PIt \FSt } } \\
=&~ \AIb \AJb \bCES{S^2}{S}
  + \rB{\AIb \AJa + \AIa \AJb} \bCES{S}{S}
\nonumber \\ &~~
  + \rB{\AIb \AJc + \AIc \AJb} \FRc \bCES{S^*}{S^*} + \AIa \AJa 
\nonumber \\ &~~
  + \rB{\AIc \AJa + \AIa \AJc + \SGij - \AIb \AJb \SGS} \FRc
  + \AIc \AJc \FRd ,
\end{align*}
thus \eqref{eq:A3} is obtained. By setting $j=i$, \eqref{eq:A2} is directly implied from \eqref{eq:A3}.
\end{proof}

\end{document}